\documentclass{article}%
\usepackage{amsmath}
\usepackage{amsfonts}
\usepackage{amssymb}
\usepackage{graphicx}
\usepackage[colorlinks,linkcolor={blue},  bookmarks={false},
pdfpagelayout={SinglePage}  ]{hyperref}%
\newtheorem{theorem}{Theorem}

\newtheorem{construction}[theorem]{Construction}

\newtheorem{definition}[theorem]{Definition}

\newtheorem{lemma}[theorem]{Lemma}

\newtheorem{proposition}[theorem]{Proposition}
\newtheorem{remark}[theorem]{Remark}

\newenvironment{proof}[1][Proof]{\noindent\textbf{#1.} }{\ \rule{0.5em}{0.5em}}

\usepackage{tikz}

\usetikzlibrary{patterns}

\newlength{\leveldistance}
\tikzstyle{fit level distance to}=[execute at begin picture={\settowidth{\leveldistance}{#1}\addtolength{\leveldistance}{1.5in}},execute at end picture={\setlength{\leveldistance}{1.5in}}]
\tikzstyle{override level distance}=[execute at begin picture={\setlength{\leveldistance}{#1}},execute at end picture={\setlength{\leveldistance}{1.5in}}]

\newlength{\defaultsiblingdistance} 
\newlength{\siblingdistance} 
\tikzstyle{max label lines}=[execute at begin picture={\setlength{\siblingdistance}{#1\baselineskip}\addtolength{\siblingdistance}{\baselineskip}}, execute at end picture={\setlength{\siblingdistance}{\defaultsiblingdistance}}]

\newcommand{\defaultsiblings}{2}
\newcommand{\siblings}{\defaultsiblings}
\tikzstyle{siblings}=[execute at begin picture={\renewcommand{\siblings}{#1}}, execute at end picture={\renewcommand{\siblings}{\defaultsiblings}}]
\tikzstyle{binary}=[siblings=2]

\tikzstyle{model}=[every node/.style={draw},grow'=right,->,
level/.style={level distance=\leveldistance, sibling distance=\siblingdistance}]
\tikzstyle{twostep}=[level 1/.style={level distance=\leveldistance, sibling distance=\siblings\siblingdistance}]

\tikzstyle{axis}=[gray]
\tikzstyle{axislabel}=[black]
\tikzstyle{tick}=[gray]
\tikzstyle{ticklabel}=[black]

\newlength{\figurewidth}
\newlength{\figureheight}
\newcommand{\twovector}[2]{\left(#1,#2\right)}

\newcommand{\conv}{\operatorname*{conv}}

\usepackage{calc}
\newcommand{\longestterm}{}
\newcommand{\setlongestterm}[1]{\renewcommand{\longestterm}{$\displaystyle#1$}}
\newcommand{\adjusttolongestterm}[2][]{\makebox[\widthof{\longestterm}#1][l]{$\displaystyle#2$}}

\newlength{\innotis}
\let\marginparnew=\marginpar
\long\def\marginpar#1{\marginparnew{\small #1}}

\begin{document}

\title{Game options with gradual exercise and cancellation under proportional
transaction costs}
\author{Alet Roux and Tomasz Zastawniak\\Department of Mathematics, University of York,\\Heslington, York YO10 5DD, United Kingdom}
\date{7 December 2016}
\maketitle

\begin{abstract}
Game (Israeli) options in a multi-asset market model with proportional transaction costs are
studied in the case when the buyer is allowed to exercise the option and the seller has the right to cancel the option gradually at a mixed (or randomised) stopping time, rather than instantly at an ordinary stopping time. Allowing gradual exercise and cancellation leads to increased flexibility in hedging, and hence tighter bounds on the option price as compared to the case of instantaneous exercise and cancellation. Algorithmic constructions for the bid and ask
prices, and the associated superhedging strategies and optimal mixed stopping times for
both exercise and cancellation are developed and illustrated. Probabilistic dual
representations for bid and ask prices are also established.
\end{abstract}

\section{Introduction}

A game (i.e.\ Israeli) option is a contract between an option buyer and seller, which allows the buyer the right to exercise the option, and the seller the right to cancel the option at any time up to expiry. The payoff associated with such a game option is due at the earliest of the exercise and cancellation times. If the option is cancelled before it is exercised, then the buyer also receives additional compensation from the seller. Game options were first introduced by Kifer~\cite{kiefer2000} and
have been studied in a frictionless setting in a number of papers; for a
survey of this work see Kifer~\cite{kifer2013a}.
Game options have proved to be important not only in their own right, but also because they underpin the theory for other traded derivatives such as convertible bonds or callable options; see e.g.\ Kallsen and K\"uhn \cite{KalKuh2006}, K\"uhn and Kyprianou \cite{KuhKyp2007}, Bielecki et~al. \cite{Biel2008}, Wang and Jin \cite{WanJin2009}, or Kwok \cite{Kwok2014}.

Transaction costs were first
considered in the context of game options by Kifer \cite{kifer2013}, who extended the results established for American options by Roux and
Zastawniak~\cite{rouxzastawniak2009} in the case of a market with a single
risky security. Kifer's work \cite{kifer2013} has recently been generalised by
Roux~\cite{roux2016} for game options in discrete multi-asset models with
proportional transaction costs.
Due to a negative result by Dolinsky \cite{Dol2013} that the superreplication price of a game option in continuous time under proportional transaction costs is the initial value of a trivial buy-and-hold strategy, both Kifer \cite{kifer2013} and Roux \cite{roux2016} study game options in discrete time. This approach is also adopted in the present paper.

Consistently with the wider literature on game options, the papers by Kifer \cite{kifer2013} and Roux \cite{roux2016} take it for granted that the option can only be exercised or cancelled instantaneously and in full, in other words, at an ordinary stopping time. This means that pricing and hedging involves non-convex optimization problems for both the buyer and seller. In this case, Kifer \cite{kifer2013} and Roux \cite{roux2016} showed that the bid and ask prices can be computed algorithmically, as can optimal strategies for both the buyer and the seller. Moreover, they established probabilistic dual representations for the bid and ask prices. In common with American options in this setting, the dual representations involve so-called mixed (or randomised) stopping times (used before in various contexts by Baxter and Chacon \cite{BaxCha1977}, Chalasani and Jha \cite{ChaJha2001}, Bouchard and Temam \cite{BouTem2005} and many others).

In the present paper we allow increased flexibility for both the buyer and seller by permitting both exercise and cancellation to take place gradually, i.e.\ at a mixed stopping time, rather than instantaneously at an ordinary stopping time. Such flexibility is available to investors who hold a portfolio of options and may wish to manage their exposure by exercising or cancelling some of these options at different times.

In the presence of proportional transaction costs, gradual exercise and cancellation is closely linked to the notion of deferred solvency; this has already been studied in the context of American options with gradual exercise by Roux and Zastawniak~\cite{rouxzastawniak2014}. In the presence of a large bid-ask spread on the underlying assets, for example in the event of temporary illiquidity in the market, an agent may become insolvent in the traditional sense at some time instant~$t$, but still able to return to solvency at a later time by trading in a self-financing way. Allowing such deferred solvency positions, rather than insisting on immediate solvency at all times, also leads to increased flexibility in constructing hedging strategies for both the seller and buyer of a game option.

In this setting, i.e.\ for game options with gradual exercise and cancellation under transaction costs and deferred solvency, we establish algorithmic constructions of the bid and ask prices and of optimal hedging strategies for both the seller and buyer of the option. In doing so, we extend the results of Kifer \cite{kifer2013} and Roux \cite{roux2016}, which apply to game options that allow only instantaneous exercise and cancellation with immediate solvency.

It turns out that, in the presence of proportional transaction costs, allowing deferred solvency along with gradual exercise and cancellation for game options leads to tighter bid-ask spreads as compared to the case of instantaneous exercise and cancellation, an advantage for the parties on either side of the option contact. Moreover, it is important to note that allowing gradual exercise and cancellation turns pricing and hedging for the buyer and seller of a game option into convex optimization problems, massively enhancing the efficiency of the pricing and hedging algorithms as compared with the non-convex case studied by Kifer \cite{kifer2013} and Roux \cite{roux2016}. Furthermore, convexity facilitates the use of duality methods, and could potentially allow extending the linear vector optimisation techniques which were developed by L\"ohne and Rudloff \cite{LohRud2014} for European options.

The methods and results presented in this paper build on a large body of work for European and American options under transaction costs, including papers by Merton \cite{Mer1989}, Dermody and Rockafellar \cite{DerRoc1991}, Boyle and Vorst \cite{BoyVor1992}, Bensaid, Lesne, Pag\`es and Scheinkman \cite{Bens1992}, Edirisinghe, Naik and Uppal \cite{EdiNaiUpp1993}, Jouini and Kallal \cite{JouKal1995}, Kusuoka \cite{Kus1995}, Koehl, Pham and Touzi \cite{KoePhaTou1999, KoePhaTou2002}, Stettner \cite{Ste1997, Ste2000}, Perrakis and Lefoll \cite{PerLef1997}, Rutkowski \cite{Rut1998}, Touzi \cite{Tou1999}, Kabanov \cite{kabanov1999}, Jouini \cite{Jou2000}, Palmer \cite{Pal2001}, Chalasani and Jha \cite{ChaJha2001}, Kabanov and Stricker \cite{kabanovstricker2001b}, Koci\'nski \cite{Koc2004}, Schachermayer \cite{schachermayer2004}, Bouchard and Temam \cite{BouTem2005}, Tokarz and Zastawniak \cite{TokZas2006}, Chen, Palmer and Sheu \cite{ChePalShe2008}, Roux, Tokarz and Zastawniak \cite{RouTokZas2008}, Roux and Zastawniak \cite{rouxzastawniak2009, rouxzastawniak2014, rouxzastawniak2016}, L\"ohne and Rudloff \cite{LohRud2014}.

The paper proceeds as follows. Section~\ref{sec:2} introduces the multi-asset model with proportional transaction costs and summarizes the results from Roux and Zastawniak~\cite{rouxzastawniak2014} that will be used in this paper. Game options with gradual exercise and cancellation are introduced in Section~\ref{sec:3}. Pricing algorithms for the seller and buyer, together with dual representations, are presented in Sections~\ref{sec:4} and~\ref{sec:5}, with proofs deferred to Section~\ref{sec:7}. An example is provided in Section~\ref{sec:6}.

\section{Preliminaries}\label{sec:2}

\subsection{Many-asset model with proportional transaction costs}

We consider the discrete-time market model with many assets (conveniently
thought of as currencies) and proportional transaction costs introduced
by Kabanov~\cite{kabanov1999} and studied by Kabanov and Stricker~\cite{kabanovstricker2001b},
Schachermayer~\cite{schachermayer2004} and others.

Let $(\Omega,\mathcal{F},\mathbb{P})$ be a probability space with
filtration~$(\mathcal{F}_{t})_{t=0}^{T}$. We assume that~$\Omega$ is finite,
$\mathcal{F}_{0}=\{\varnothing,\Omega\}$, $\mathcal{F}_{T}=\mathcal{F}%
=2^{\Omega}$ and
$\mathbb{P}(\omega)>0$ for all $\omega\in\Omega$. For each $t=0,\ldots,T$,
by~$\Omega_{t}$ we denote the collection of atoms of~$\mathcal{F}_{t}$, called
the \emph{nodes} of the associated tree model.

The market model contains~$d$ assets or currencies. At each trading date
$t=0,1,\ldots,T$ and for each $k,j=1,\ldots,d$, one unit of asset~$k$ can be
obtained by exchanging $\pi_{t}^{jk}>0$ units of asset~$j$. We assume that the
exchange rates~$\pi_{t}^{jk}$ are $\mathcal{F}_{t}$-measurable and $\pi
_{t}^{jj}=1$ for all $t$ and $j,k$.

For each $t=0,\ldots,T$ let $\mathcal{L}_{t}:=\mathcal{L}^{0}(\mathbb{R}%
^{d};\mathcal{F}_{t})$ be the collection of $\mathcal{F}_{t}$-measurable
$\mathbb{R}^{d}$-valued random variables. We can identify elements
of~$\mathcal{L}_{t}$ with $\mathbb{R}^{d}$-valued functions on~$\Omega_{t}$.
Any $x\in\mathcal{L}_{t}$ can be thought of as a portfolio with positions
$x^{1},\ldots,x^{d}$ in the $d$~assets. We say that a portfolio $x\in
\mathcal{L}_{t}$ can be \emph{exchanged} into a portfolio $y\in
\mathcal{L}_{t}$ at time~$t$ whenever there are $\mathcal{F}_{t}$-measurable
random variables $\beta^{jk}\geq0$, $j,k=1,\ldots,d$ such that for all
$k=1,\ldots,d$
\[
y^{k}=x^{k}+\sum_{j=1}^{d}\beta^{jk}-\sum_{j=1}^{d}\beta^{kj}\pi_{t}^{kj},
\]
where $\beta^{jk}$ represents the number of units of asset~$k$ received as a
result of exchanging some units of asset~$j$.

The \emph{solvency cone} $\mathcal{K}_{t}\subseteq\mathcal{L}_{t}$ is the set
of portfolios that are \emph{solvent} at time~$t$, i.e.\ those portfolios at
time~$t$ that can be exchanged into portfolios with non-negative positions in
all~$d$ assets. It follows that $\mathcal{K}_{t}$ is the polyhedral convex
cone generated by the canonical basis $e^{1},\ldots,e^{d}$ of $\mathbb{R}^{d}$
and the vectors $\pi_{t}^{jk}e^{j}-e^{k}$ for $j,k=1,\ldots,d$. We also refer
to~$\mathcal{K}_{t}$ as the \emph{immediate solvency cone} to distinguish it
from the so-called \emph{deferred solvency cone}~$\mathcal{Q}_{t}$ to be
introduced later.

A \emph{trading strategy} $y=(y_{t})_{t=0}^{T+1}$ is a predictable
$\mathbb{R}^{d}$-valued process with final value assumed to be $y_{T+1}=0$ for
notational convenience. For each $t>0$ the portfolio $y_{t}\in\mathcal{L}%
_{t-1}$ is held from time $t-1$ to time $t$, and~$y_{0}$ is the initial
endowment. We denote by~$\Phi$ the set of such trading strategies.

A trading strategy $y\in\Phi$ is said to be \emph{self-financing} whenever
$y_{t}-y_{t+1}\in\mathcal{K}_{t}$ for all $t=0,\ldots,T-1$. Note that no
implicitly assumed self-financing condition is included in the definition
of~$\Phi$.

A trading strategy $y\in\Phi$ is called an \emph{arbitrage opportunity} if it
is self-financing, $y_{0}=0$ and there is a portfolio $x\in\mathcal{L}%
_{T}\setminus\{0\}$ with $x^{j}\geq0$ for each $j=1,\ldots,d$ and such that
$y_{T}-x\in\mathcal{K}_{T}$. This notion of arbitrage was considered
by~\cite{schachermayer2004}, and its absence is formally different but
equivalent to the weak no-arbitrage condition introduced
by~\cite{kabanovstricker2001b}.

\begin{theorem}
[\cite{kabanovstricker2001b,schachermayer2004}]\label{th:2012-10-03:ftap} The
model admits no arbitrage opportunity if and only if there exists a
probability measure~$\mathbb{Q}$ equivalent to~$\mathbb{P}$ and an
$\mathbb{R}^{d}$-valued $\mathbb{Q}$-martingale $S=(S_{t})_{t=0}^{T}$ such
that
\begin{equation}
S_{t}\in\mathcal{K}_{t}^{\ast}\setminus\{0\}\text{\quad for all }t=0,\ldots,T,
\label{eq:th:2012-10-03:ftap}%
\end{equation}
where $\mathcal{K}_{t}^{\ast}:=\{y\in\mathcal{L}_{t}\,|\,y\cdot x\geq0$ for
all $x\in\mathcal{K}_{t}\}$ is the polar of $-\mathcal{K}_{t}$.
\end{theorem}

We denote by~$\mathcal{P}$ the set of pairs $(\mathbb{Q},S)$ satisfying the
conditions in Theorem~\ref{th:2012-10-03:ftap}, and by~$\bar{\mathcal{P}}$ the
set of pairs $(\mathbb{Q},S)$ satisfying the conditions in
Theorem~\ref{th:2012-10-03:ftap} but with~$\mathbb{Q}$ absolutely continuous
with respect to (and not necessarily equivalent to)~$\mathbb{P}$. We assume
for the remainder of this paper that the model admits no arbitrage
opportunities, i.e.~$\mathcal{P}\neq\varnothing$.

Any portfolio~$x\in\mathcal{K}_t$ is immediately solvent at time~$t$, in the sense that it can be converted at time $t$ into one with non-negative positions in all~$d$ assets.
For American and
game options under transaction costs, the following weaker type of solvency
also proves useful. We denote by~$\mathcal{Q}_{t}$ the collection of
portfolios $x\in\mathcal{L}_{t}$ such that there is a sequence $y_{s}%
\in\mathcal{L}_{s-1}$ for $s=t+1,\ldots,T+1$ satisfying the conditions
\[
x-y_{t+1}\in\mathcal{K}_{t},\quad y_{s}-y_{s+1}\in\mathcal{K}_{s}\text{ for
all }s=t+1,\ldots,T,\quad y_{T+1}=0.
\]
We call such a sequence $y_{t+1},\ldots,y_{T+1}$ a \emph{liquidation strategy}
starting from~$x$ at time~$t$.

The portfolios in~$\mathcal{Q}_{t}$ are
those that can eventually (though possibly not immediately at time~$t$) be
converted by means of a sequence of self-financing transactions into a
portfolio with non-negative positions in all assets. An equivalent way of
constructing the deferred solvency cones is to put
\[
\mathcal{Q}_{T}:=\mathcal{K}_{T}%
\]
and then
\[
Q_{t}:=\mathcal{Q}_{t+1}\cap\mathcal{L}_{t}+\mathcal{K}_{t}\text{ for
}t=T-1,\ldots,0
\]
by backward induction. It turns out that~$\mathcal{Q}_{t}$ is a convex polyhedral cone. We call it
the \emph{deferred solvency cone}. See~\cite{rouxzastawniak2014} for more information on
deferred solvency.

\subsection{Mixed stopping times}

A \emph{mixed} (or \emph{randomised}) \emph{stopping time} is a non-negative adapted process $\phi=(\phi_{t})_{t=0}^{T}$ such that
\[
\sum_{t=0}^{T}\phi_{t}=1.
\]
The collection of mixed stopping times will be denoted by~$\mathcal{X}$.

For any $\phi\in\mathcal{X}$ we put%
\begin{equation}
\phi_{t}^{\ast}:=\sum_{s=t}^{T}\phi_{s}\text{ for }t=0,\ldots,T,\quad
\text{\quad}\phi_{T+1}^{\ast}:=0. \label{eq:hfh456wfd648d}%
\end{equation}
Observe that~$\phi^{\ast}$ is a predictable process since $\phi_{0}^{\ast}=1$
is $\mathcal{F}_{0}$-measurable and $\phi_{t}^{\ast}=1-\sum_{s=0}^{t-1}%
\phi_{s}$ is $\mathcal{F}_{t-1}$-measurable for each $t=1,\ldots
,T$.

For example, in the case of a game option subject to gradual cancellation, $\phi_t$~could represent a fraction of the option that is cancelled at time~$t$, whereas~$\phi^\ast_t$ would be the part of the option that has not been cancelled before~$t$.

For any adapted process~$X$ and for any $\phi\in\mathcal{X}$ we
define the process~$X$ \emph{evaluated at}~$\phi$ as%
\[
X_{\phi}:=\sum_{t=0}^{T}\phi_{t}X_{t}.
\]
We also put
\[
X_{t}^{\phi\ast}:=\sum_{s=t}^{T}\phi_{s}X_{s}\text{ for }t=0,\ldots
,T,\quad\quad X_{T+1}^{\phi\ast}:=0.
\]

The collection~$\mathcal{T}$ of ordinary stopping times can be embedded
in~$\mathcal{X}$ by identifying every $\tau\in\mathcal{T}$ with the mixed
stopping time $\chi^{\tau}\in\mathcal{X}$ defined as%
\[
\chi_{t}^{\tau}:=\mathbf{1}_{\left\{  t=\tau\right\}  }%
\]
for each $t=0,\ldots,T$. (Here $\mathbf{1}_{A}$ denotes the indicator function
of an event $A\in\mathcal{F}$.)

\subsection{American options with gradual exercise and cancellation}

Here we collect the main notions and results concerning American options with
gradual exercise and cancellation under proportional transaction costs; for
full details, see~\cite{rouxzastawniak2014}. These will be extended to game
options, and will also be used as tools to establish some key results for this extension.

Consider an American option with adapted payoff process~$Z=(Z_{t})_{t=0}^{T}$,
where $Z_{t}\in\mathcal{L}_{t}$ represents a portfolio of~$d$ assets for
each~$t=0,\ldots,T$. If the buyer of the option is allowed to exercise it
gradually according to a mixed stopping time $\psi\in\mathcal{X}$, then the
sequence of portfolios~$\psi_{t}Z_{t}$ is to be delivered by the option seller
to the buyer at times $t=0,\ldots,T$.

The seller needs to hedge against all mixed stopping times $\psi\in
\mathcal{X}$ that can be chosen by the buyer. Because the seller can react to
the buyer's choice of~$\psi$, the hedging strategy may depend on~$\psi$. We
are going to write~$z_{t}^{\psi}$ for the time~$t$ position in the strategy.

On each trading date~$t$, the seller needs to deliver the payoff~$\psi
_{t}Z_{t}$ and to rebalance the strategy from $z_{t}^{\psi}$ to~$z_{t+1}%
^{\psi}$ without injecting any additional wealth, and can only use knowledge of~$\psi$ and the market up to and including time$~t$. This leads to
the following conditions.

\begin{definition}
\upshape Let $Z=(Z_{t})_{t=0}^{T}$ be an adapted process. For an American
option with payoff process$~Z$ and gradual exercise, a \emph{seller's
superhedging strategy} is a mapping $z:\mathcal{X}\rightarrow\Phi$ that
satisfies the \emph{rebalancing} condition%
\begin{equation}
\forall\psi\in\mathcal{X~}\forall t=0,\ldots,T:~z_{t}^{\psi}-\psi_{t}%
Z_{t}-z_{t+1}^{\psi}\in\mathcal{K}_{t} \label{eq:jgt775g10}%
\end{equation}
and the \emph{non-anticipation} condition%
\begin{equation}
\forall\psi,\psi^{\prime}\in\mathcal{X}~\forall t=0,\ldots
,T:~\textstyle\bigcap_{s=0}^{t-1}\{\psi_{s}=\psi_{s}^{\prime}\}\subseteq
\{z_{t}^{\psi}=z_{t}^{\psi^{\prime}}\}. \label{eq:dg453hd7ah}%
\end{equation}
The family of such strategies will be denoted by~$\Psi^{\mathrm{a}}(Z)$.
\end{definition}

\begin{definition}
\label{Def:hfy468ahd564vk}\upshape The \emph{seller's }(or \emph{ask})\emph{
price}\ in currency $j=1,\ldots,d$ of an American option with payoff
process$~Z$ and gradual exercise is defined as%
\[
p_{j}^{\mathrm{a}}(Z):=\inf\left\{  x\in\mathbb{R}\,|\,\exists z\in
\Psi^{\mathrm{a}}(Z):xe^{j}=z_{0}\right\}  .
\]

\end{definition}

The following representation of the seller's price was obtained
in~\cite{rouxzastawniak2014}. In this representation, for any $\psi
\in\mathcal{X}$, we denote by$~\mathcal{\bar{P}}_{j}^{\mathrm{d}}(\psi)$ the
collection of pairs $(\mathbb{Q},S)$ such that~$\mathbb{Q}$ is a probability
measure absolutely continuous with respect to~$\mathbb{P}$ and~$S$ is an
$\mathbb{R}^{d}$-valued adapted process such that%
\[
S_{t}\in\mathcal{Q}_{t}^{\ast}\setminus\{0\}\quad\text{and}\quad
\mathbb{E}_{\mathbb{Q}}(S_{t+1}^{\psi\ast}|\mathcal{F}_{t})\in\mathcal{Q}%
_{t}^{\ast}\quad\text{for all }t=0,\ldots,T,
\]
where $\mathcal{Q}_{t}^{\ast}:=\{y\in\mathcal{L}_{t}\,|\,y\cdot x\geq0$ for
all $x\in\mathcal{Q}_{t}\}$ is the polar of $-\mathcal{Q}_{t}$.

\begin{theorem}
[{\cite[Theorem~4.2]{rouxzastawniak2014}}]\label{Thm:jjf658wgffvdsdf85}The
seller's price in currency $j=1,\ldots,d$ of an American option with payoff
process~$Z$ and gradual exercise can be expressed as%
\[
p_{j}^{\mathrm{a}}(Z)=\max_{\psi\in\mathcal{X}}\max_{(\mathbb{Q}%
,S)\in\mathcal{\bar{P}}_{j}^{\mathrm{d}}(\psi)}\mathbb{E}_{\mathbb{Q}}((Z\cdot
S)_{\psi}).
\]

\end{theorem}

\section{Game options with gradual exercise and cancellation}\label{sec:3}

A game option as introduced in~\cite{kiefer2000} and studied in
\cite{roux2016} is a contract between an option buyer and seller, which gives
the buyer the right to exercise the option at any stopping time $\tau
\in\mathcal{T}$, and also gives the seller the right to cancel the option at
any stopping time $\sigma\in\mathcal{T}$. There are two adapted processes
$Y=(Y_{t})_{t=0}^{T}$ and $X=(X_{t})_{t=0}^{T}$ which determine, respectively,
the payoffs due when exercising and cancelling the option. In the presence of
transaction costs $Y$ and~$X$ are $\mathbb{R}^{d}$-valued
(i.e.\ portfolio-valued) processes; see~\cite{kifer2013} in the case when
$d=2$ and~\cite{roux2016} for any $d\geq2$. The seller has to deliver the
portfolio~$Y_{\tau}$ to the buyer at time~$\tau$ when $\sigma\geq\tau$ or the
portfolio~$X_{\sigma}$ at time~$\sigma$ when \mbox{$\sigma<\tau$}. That is,\ the
option will be terminated at time $\sigma\wedge\tau$ due to exercise or
cancellation, and the portfolio%
\[
Q_{\sigma,\tau}:=\mathbf{1}_{\left\{  \sigma\geq\tau\right\}  }Y_{\tau
}+\mathbf{1}_{\left\{  \sigma<\tau\right\}  }X_{\sigma}%
\]
will be changing hands at that time. Observe that exercising the option takes
priority over cancellation when $\sigma=\tau$. Additionally, it is assumed
that%
\begin{equation}
X_{t}-Y_{t}\in\mathcal{K}_{t}\quad\text{for each }t=0,\ldots,T.
\label{eq:hhjfd6487shd6543}%
\end{equation}
The difference $X_{t}-Y_{t}$ can be regarded as a penalty payable by the
seller on top of the payoff~$Y_{t}$ when cancelling the option. We shall refer
to an option of this kind as a \emph{game }(or \emph{Israeli})\emph{ option
with instant exercise and cancellation}, to distinguish it from one with
gradual exercise and cancellation as described below.

In the present work we allow both the buyer and seller the freedom to exercise
or, respectively, to cancel the option gradually according to mixed stopping
times. If the buyer chooses a mixed stopping time $\psi\in\mathcal{X}$ as the
exercise time and the seller selects a mixed stopping time $\phi\in
\mathcal{X}$ to be the cancellation time, then on each trading date
$t=0,\ldots,T$ the buyer will first be exercising a fraction $\psi_{t}%
/\psi_{t}^{\ast}$ of the current position in the option, and then the seller
will be cancelling a fraction $\phi_{t}/\phi_{t}^{\ast}$ of the remaining
position in the option, where $\psi_{t}^{\ast}$ and~$\phi_{t}^{\ast}$ are
given by~(\ref{eq:hfh456wfd648d}). Once again, exercising takes priority
over cancellation.

In these circumstances, starting with an initial position of $\psi_{0}^{\ast
}\phi_{0}^{\ast}=1$~option at time~$0$, we are going to show by induction that
$\psi_{t}^{\ast}\phi_{t}^{\ast}$ of the option will neither be exercised nor
cancelled before time~$t$, for each $t=0,\ldots,T$. It means that $\psi
_{t}/\psi_{t}^{\ast}$ of the current position $\psi_{t}^{\ast}\phi_{t}^{\ast}%
$, that is, $\left(  \psi_{t}/\psi_{t}^{\ast}\right)  \psi_{t}^{\ast}\phi
_{t}^{\ast}=\psi_{t}\phi_{t}^{\ast}$ of the option will be exercised at~$t$,
given that the buyer has priority to exercise. The remaining position in the
option will then be $\psi_{t}^{\ast}\phi_{t}^{\ast}-\psi_{t}\phi_{t}^{\ast
}=\left(  \psi_{t}^{\ast}-\psi_{t}\right)  \phi_{t}^{\ast}=\psi_{t+1}^{\ast
}\phi_{t}^{\ast}$, hence $\left(  \phi_{t}/\phi_{t}^{\ast}\right)  \psi
_{t+1}^{\ast}\phi_{t}^{\ast}=\psi_{t+1}^{\ast}\phi_{t}$ of the option will be
cancelled by the seller at~$t$. Altogether, $\psi_{t}\phi_{t}^{\ast}%
+\psi_{t+1}^{\ast}\phi_{t}$ of the option will be terminated at~$t$ due to
exercise or cancellation, leaving
\begin{align*}
\psi_{t}^{\ast}\phi_{t}^{\ast}-\psi_{t}\phi_{t}^{\ast}-\psi_{t+1}^{\ast}%
\phi_{t}  &  =\left(  \psi_{t}+\psi_{t+1}^{\ast}\right)  \left(  \phi_{t}%
+\phi_{t+1}^{\ast}\right)  -\psi_{t}\left(  \phi_{t}+\phi_{t+1}^{\ast}\right)
-\psi_{t+1}^{\ast}\phi_{t}\\
&  =\psi_{t+1}^{\ast}\phi_{t+1}^{\ast}%
\end{align*}
of the option neither exercised nor cancelled before or at~$t$, to be carried
forward to time~$t+1$. This completes the induction.

\begin{remark}
\upshape The \emph{minimum} $\psi\wedge\phi$
of mixed stopping times $\psi,\phi\in\mathcal{X}$ can be defined as%
\[
(\psi\wedge\phi)_{t}:=\psi_{t}\phi_{t}^{\ast}+\psi_{t+1}^{\ast}\phi_{t}%
\]
for each $t=0,\ldots,T$; see~\cite{kifer2013}. The above argument shows that a
game option with gradual exercise and cancellation will be terminated
according to the mixed stopping time $\psi\wedge\phi$.
\end{remark}

On each trading date $t=0,\ldots,T$, since $\psi_{t}\phi_{t}^{\ast}$ of the
option is to be exercised and $\psi_{t+1}^{\ast}\phi_{t}$ of the option to be
cancelled, the seller will be delivering to the buyer the portfolio%
\[
G_{t}^{\phi,\psi}:=\psi_{t}\phi_{t}^{\ast}Y_{t}+\psi_{t+1}^{\ast}\phi_{t}%
X_{t},
\]
where $Y=(Y_{t})_{t=0}^{T}$ and $X=(X_{t})_{t=0}^{T}$ are the exercise and
cancellation processes characterising the game option, that is, $\mathbb{R}%
^{d}$-valued adapted processes that satisfy~(\ref{eq:hhjfd6487shd6543}).
Clearly, $G^{\phi,\psi}=(G_{t}^{\phi,\psi})_{t=0}^{T}$ is an $\mathbb{R}^{d}%
$-valued adapted process, which we shall be referring to as the \emph{payoff
process} for the game option.

\begin{definition}
\upshape A \emph{game }(or\emph{ Israeli})\emph{ option }$(Y,X)$ \emph{with
gradual exercise and cancellation} is a derivative security that can be
exercised according to a mixed stopping time $\psi\in\mathcal{X}$ chosen by
the buyer or cancelled according to a mixed stopping time $\phi\in\mathcal{X}$
chosen by the seller, giving the buyer the right to receive and obliging the
seller to deliver the portfolio~$G_{t}^{\phi,\psi}=\psi_{t}\phi_{t}^{\ast
}Y_{t}+\psi_{t+1}^{\ast}\phi_{t}X_{t}$ on each trading date $t=0,\ldots,T$.
\end{definition}

\begin{remark}
\upshape In contrast to the above payoff process~$G_{t}^{\phi,\psi}$,
Kifer~\cite{kifer2013} refers to the random variable
\[
Q_{\phi,\psi}:=\sum_{s=0}^{T}\sum_{t=0}^{T}\phi_{s}\psi_{t}Q_{s,t}%
\]
as the `payoff' of a game option with exercise and cancellation according to
mixed stopping times~$\phi,\psi\in\mathcal{X}$, without specifying the time
instant when this portfolio should be changing hands. However, the payoff of
such an option should not be a single random variable but in fact an adapted
process representing the flow of portfolios to be delivered on each trading
date $t=0,\ldots,T$. We observe that%
\[
Q_{\phi,\psi}=\sum_{t=0}^{T}G_{t}^{\phi,\psi},
\]
i.e.\ $Q_{\phi,\psi}$ happens to be the total of all the~$G_{t}^{\phi,\psi}$
for $t=0,\ldots,T$.

In the present paper~$Q_{\phi,\psi}$ will prove useful in a different role.
Namely, identifying the mixed stopping time~$\chi^{t}\in\mathcal{X}$ with a
deterministic time~$t$, we are going to use $Q_{\phi,\chi^{t}}$ for
$t=0,\ldots,T$ as the payoff process of an American option with gradual
exercise and invoke the results of~\cite{rouxzastawniak2014} to establish a
probabilistic representation of the seller's price for a game option under
gradual exercise and cancellation; see Lemma~\ref{Lem:8f5sgdjf} and
Theorem~\ref{Thm:eudn9y65}. Similarly, in the buyer's
case, we are going to use an American option with gradual exercise and payoff
process $-Q_{\chi^{t},\psi}$ for $t=0,\ldots,T$; see Lemma~\ref{Lem:k856dg3a}
and Theorem~\ref{Thm:dhhd5436syy}.
\end{remark}

\section{Seller's price and superhedging strategies}\label{sec:4}

The seller of a game option $(Y,X)$ with gradual exercise and cancellation
needs to hedge against any mixed stopping time $\psi\in\mathcal{X}$ chosen by
the buyer to exercise the option. The seller can do this by following a
trading strategy \mbox{$u^{\psi}=(u_{t}^{\psi})_{t=0}^{T}\in\Phi$}, which may depend
on~$\psi$. Since~$u_{t}^{\psi}$ denotes a portfolio held over time step~$t$,
that is, between times $t-1$ and~$t$, it follows that~$u_{t}^{\psi}$ may
depend on the values $\psi_{0},\ldots,\psi_{t-1}$ known to the seller at
time~$t-1$, when this portfolio is to be created, but not on the yet unknown
(to the seller) values $\psi_{t},\ldots,\psi_{T}$. This is the reason for the
non-anticipation condition~(\ref{eq:ug75hsdj6}) in
Definition~\ref{Def:dn57sn122}.

In addition to choosing the trading strategy $u^{\psi}\in\Phi$, the seller can
select a mixed stopping time $\phi\in\mathcal{X}$ to cancel the option, and must be able to deliver
the portfolio~$G_{t}^{\phi,\psi}$ on each date $t=0,\ldots,T$ without
injecting any additional wealth into the strategy. This justifies the
rebalancing condition~(\ref{eq:f7enasg0}).

\begin{definition}
\label{Def:dn57sn122}\upshape For a game option $(Y,X)$ with gradual exercise
and cancellation, a \emph{seller's superhedging strategy }is a pair $(\phi
,u)$, where $\phi\in\mathcal{X}$ and $u:\mathcal{X}\rightarrow\Phi$, that
satisfies the \emph{rebalancing} condition%
\begin{equation}
\forall\psi\in\mathcal{X~}\forall t=0,\ldots,T:~u_{t}^{\psi}-G_{t}^{\phi,\psi
}-u_{t+1}^{\psi}\in\mathcal{K}_{t} \label{eq:f7enasg0}%
\end{equation}
and the \emph{non-anticipation} condition%
\begin{equation}
\forall\psi,\psi^{\prime}\in\mathcal{X}~\forall t=0,\ldots
,T:~\textstyle\bigcap_{s=0}^{t-1}\{\psi_{s}=\psi_{s}^{\prime}\}\subseteq
\{u_{t}^{\psi}=u_{t}^{\psi^{\prime}}\}. \label{eq:ug75hsdj6}%
\end{equation}
The family of such strategies will be denoted by~$\Phi^{\mathrm{a}}(Y,X)$.
\end{definition}

The least expensive (in a particular currency~$j$) seller's superhedging
strategy gives rise to the seller's price of the option.

\begin{definition}
\label{Def:nnf65uwsu48}\upshape The \emph{seller's }(or \emph{ask})\emph{
price} in currency $j=1,\ldots,d$ of a game option $(Y,X)$ with gradual
exercise and cancellation is defined as%
\[
\pi_{j}^{\mathrm{a}}(Y,X):=\inf\left\{  x\in\mathbb{R}\,|\,\exists(\phi
,u)\in\Phi^{\mathrm{a}}(Y,X):xe^{j}=u_{0}\right\}  .
\]

\end{definition}

\subsection{Seller's pricing algorithm}

The following is an iterative construction of the set of initial endowments
that allow superhedging the seller's position in a game option with gradual
exercise and cancellation.

\begin{construction}
\label{constr:seller}Construct adapted sequences
$\mathcal{Y}_{t}^{\mathrm{a}},\mathcal{X}_{t}^{\mathrm{a}},\mathcal{V}%
_{t}^{\mathrm{a}},\mathcal{W}_{t}^{\mathrm{a}},\mathcal{Z}_{t}^{\mathrm{a}}$
for $t=\break 0,\ldots,T$ as follows. First, put%
\[
\mathcal{Y}_{t}^{\mathrm{a}}:=Y_{t}+\mathcal{Q}_{t},\quad\mathcal{X}%
_{t}^{\mathrm{a}}:=X_{t}+\mathcal{Q}_{t}%
\]
for all $t=0,\ldots,T$ and%
\[
\mathcal{W}_{T}^{\mathrm{a}}:=\mathcal{V}_{T}^{\mathrm{a}}:=\mathcal{L}%
_{T},\quad\mathcal{Z}_{T}^{\mathrm{a}}:=\mathcal{Y}_{T}^{\mathrm{a}}.
\]
Then, for $t=T-1,\ldots,0$ define by backward induction%
\begin{align*}
\mathcal{W}_{t}^{\mathrm{a}}  &  :=\mathcal{Z}_{t+1}^{\mathrm{a}}%
\cap\mathcal{L}_{t},\\
\mathcal{V}_{t}^{\mathrm{a}}  &  :=\mathcal{W}_{t}^{\mathrm{a}}+\mathcal{Q}%
_{t},\\
\mathcal{Z}_{t}^{\mathrm{a}}  &  :=\operatorname*{conv}\{\mathcal{V}%
_{t}^{\mathrm{a}},\mathcal{X}_{t}^{\mathrm{a}}\}\cap\mathcal{Y}_{t}%
^{\mathrm{a}},
\end{align*}
where $\operatorname*{conv}\{\mathcal{V}_{t}^{\mathrm{a}},\mathcal{X}%
_{t}^{\mathrm{a}}\}$ is the convex hull of~$\mathcal{V}_{t}^{\mathrm{a}}$
and$~\mathcal{X}_{t}^{\mathrm{a}}$.
\end{construction}

By a similar argument as in the proof of Proposition~5.1
in~\cite{rouxzastawniak2014}, it follows that $\mathcal{Z}_{t}^{\mathrm{a}}$
are polyhedral convex sets for all~$t$. We shall see that $\mathcal{Z}%
_{0}^{\mathrm{a}}$ is the set of initial endowments that allow the seller to
superhedge their position in the game option $(Y,X)$ with gradual exercise and
cancellation. Once $\mathcal{Z}_{0}^{\mathrm{a}}$ has been constructed, the
following result can be used to obtain the seller's price of the option.

\begin{theorem}
\label{Thm:fjusab76}The seller's price in currency $j=1,\ldots,d$ of a game
option~$(Y,X)$ with gradual exercise and cancellation can be expressed as%
\[
\pi_{j}^{\mathrm{a}}(Y,X)=\min\left\{  x\in\mathbb{R}\,|\,xe^{j}\in
\mathcal{Z}_{0}^{\mathrm{a}}\right\}  .
\]

\end{theorem}

To prove this theorem, we introduce an auxiliary family $\Lambda^{\mathrm{a}%
}(Y,X)$, the elements of which can be thought of as the strategies
superhedging the seller's position in a game option with gradual cancellation,
instant (rather than gradual) exercise and deferred (rather than immediate)
solvency. The theorem and the following propositions are proved in the
Appendix, Section~\ref{Sect:Appendix}.

\begin{definition}
\upshape We define $\Lambda^{\mathrm{a}}(Y,X)$ as the family consisting of all
pairs $(\phi,z)$, where $\phi\in\mathcal{X}$ and $z\in\Phi$, that satisfy the
conditions%
\begin{align*}
z_{t}-\phi_{t}X_{t}-z_{t+1}  &  \in\mathcal{Q}_{t}\quad\text{for all
}t=0,\ldots,T-1,\\
z_{t}-\phi_{t}^{\ast}Y_{t}  &  \in\mathcal{Q}_{t}\quad\text{for all
}t=0,\ldots,T.
\end{align*}

\end{definition}

According to the next proposition, $\mathcal{Z}_{0}^{\mathrm{a}}$ coincides
with the set of initial endowments for the strategies in $\Lambda^{\mathrm{a}%
}(Y,X)$.

\begin{proposition}
\label{Prop:7djhb78a}%
\[
\mathcal{Z}_{0}^{\mathrm{a}}=\left\{  z_{0}\in\mathbb{R}^{d}\,|\,(\phi
,z)\in\Lambda^{\mathrm{a}}(Y,X)\right\}  .
\]

\end{proposition}

We also claim that the set of initial endowments for the strategies in
$\Lambda^{\mathrm{a}}(Y,X)$ coincides with that for the strategies in
$\Phi^{\mathrm{a}}(Y,X)$.

\begin{proposition}
\label{Prop:7d4hnk0a}%
\[
\left\{  z_{0}\in\mathbb{R}^{d}\,|\,(\phi,z)\in\Lambda^{\mathrm{a}%
}(Y,X)\right\}  =\left\{  u_{0}\in\mathbb{R}^{d}\,|\,(\phi,u)\in
\Phi^{\mathrm{a}}(Y,X)\right\}  .
\]

\end{proposition}

It follows from Propositions~\ref{Prop:7djhb78a} and~\ref{Prop:7d4hnk0a}
that~$\mathcal{Z}_{0}^{\mathrm{a}}$ is the family of initial endowments for
all strategies superhedging the seller's position in a game option with
gradual exercise and cancellation. This is what's needed to prove
Theorem~\ref{Thm:fjusab76}, which links the seller's price $\pi_{j}%
^{\mathrm{a}}(Y,X)$ with~$\mathcal{Z}_{0}^{\mathrm{a}}$. Full details can be
found in the Appendix, Section~\ref{Sect:Appendix}.

\subsection{Seller's price representation\label{Sect:sel_price_rep}}

In this section we obtain a dual representation of the seller's price for game
options with gradual exercise and cancellation. This relies on a similar
result established in \cite{rouxzastawniak2014} for American options with
gradual exercise; see Theorem~\ref{Thm:jjf658wgffvdsdf85}.

Observe that, by Definition~\ref{Def:nnf65uwsu48},%
\begin{align*}
\pi_{j}^{\mathrm{a}}(Y,X)  &  =\inf\left\{  x\in\mathbb{R}\,|\,\exists
(\phi,u)\in\Phi^{\mathrm{a}}(Y,X):xe^{j}=u_{0}\right\} \\
&  =\inf_{\phi\in\chi}\inf\left\{  x\in\mathbb{R}\,|\,\exists u:(\phi
,u)\in\Phi^{\mathrm{a}}(Y,X),xe^{j}=u_{0}\right\}  .
\end{align*}
Hence, as a consequence of Lemma~\ref{Lem:8f5sgdjf} below, together with
Definition~\ref{Def:hfy468ahd564vk}, we have
\begin{align*}
\pi_{j}^{\mathrm{a}}(Y,X)  &  =\inf_{\phi\in\chi}\inf\left\{  x\in
\mathbb{R}\,|\,\exists z\in\Psi^{\mathrm{a}}(Q_{\phi,\,\cdot\,}):xe^{j}%
=z_{0}\right\} \\
&  =\inf_{\phi\in\chi}p_{j}^{\mathrm{a}}(Q_{\phi,\,\cdot\,}),
\end{align*}
where $Q_{\phi,\,\cdot\,}=(Q_{\phi,t})_{t=0}^{T}$ with
\[
Q_{\phi,t}:=Q_{\phi,\chi^{t}}\text{\quad for }t=0,\ldots,T
\]
is the payoff process for an American option with gradual exercise, and where
$p_{j}^{\mathrm{a}}(Q_{\phi,\,\cdot\,})$ is the seller's price of such an American option.

\begin{lemma}
\label{Lem:8f5sgdjf}For any $\phi\in\mathcal{X}$%
\[
\left\{  u_{0}\,|\,(\phi,u)\in\Phi^{\mathrm{a}}(Y,X)\right\}  =\left\{
z_{0}\,|\,z\in\Psi^{\mathrm{a}}(Q_{\phi,\,\cdot\,})\right\}  ,
\]
where $Q_{\phi,\,\cdot\,}=(Q_{\phi,t})_{t=0}^{T}$ is the payoff process of an
American option.
\end{lemma}

The lemma is proved in the Appendix, Section~\ref{Sect:Appendix}. It turns out
that the infimum over $\phi\in\chi$ in
\[
\pi_{j}^{\mathrm{a}}(Y,X)=\inf_{\phi\in\chi}p_{j}^{\mathrm{a}}(Q_{\phi
,\,\cdot\,})
\]
is, in fact, a minimum. Moreover, $p_{j}^{\mathrm{a}}(Q_{\phi,\,\cdot\,})$ can
be represented as in Theorem~\ref{Thm:jjf658wgffvdsdf85}. This leads to the
following representation.

\begin{theorem}
\label{Thm:eudn9y65}The seller's price in currency $j=1,\ldots,d$ of a game
option~$(Y,X)$ with gradual exercise and cancellation can be represented as%
\[
\pi_{j}^{\mathrm{a}}(Y,X)=\min_{\phi\in\mathcal{X}}\max_{\psi\in\mathcal{X}%
}\max_{(\mathbb{Q},S)\in\mathcal{\bar{P}}_{j}^{\mathrm{d}}(\psi)}%
\mathbb{E}_{\mathbb{Q}}((Q_{\phi,\,\cdot\,}\cdot S)_{\psi}).
\]

\end{theorem}

The details of the proof can be found, once again, in the Appendix,
Section~\ref{Sect:Appendix}.

\section{Buyer's price and superhedging strategies}\label{sec:5}

The buyer of a game option $(Y,X)$ will be able to select a mixed stopping
time~$\psi\in\mathcal{X}$ to exercise the option, and can follow a trading
strategy $u^{\phi}=(u_{t}^{\phi})_{t=0}^{T}\in\Phi$, which may depend on the
cancellation time $\phi\in\mathcal{X}$ chosen by the seller. On each date
$t=0,\ldots,T$ the buyer will be taking delivery of the portfolio~$G_{t}%
^{\phi,\psi}$ and can rebalance the current position~$u_{t}^{\phi}$ in the
strategy into~$u_{t+1}^{\phi}$ in a self-financing way, i.e.\ without
injecting any additional wealth. The portfolio~$u_{t}^{\phi}$ created by the
buyer at time $t-1$ may depend on the seller's cancellation strategy $\phi
_{0},\ldots,\phi_{t-1}$ up to and including time $t-1$, but not on the values
$\phi_{t},\ldots,\phi_{T}$, as these will not yet be known to the buyer at
time~$t-1$. These considerations lead to the following definition.

\begin{definition}
\upshape For a game option $(Y,X)$ with gradual exercise and cancellation, a
\emph{buyer's superhedging strategy }is a pair $(\psi,u)$, where $\psi
\in\mathcal{X}$ and
$u:\mathcal{X}\rightarrow\Phi$, that satisfies the \emph{rebalancing}
condition%
\begin{equation}
\forall\phi\in\mathcal{X~}\forall t=0,\ldots,T:~u_{t}^{\phi}+G_{t}^{\phi,\psi
}-u_{t+1}^{\phi}\in\mathcal{K}_{t} \label{eq:hf65hsdn9}%
\end{equation}
and the \emph{non-anticipation} condition%
\begin{equation}
\forall\phi,\phi^{\prime}\in\mathcal{X}~\forall t=0,\ldots
,T:~\textstyle\bigcap_{s=0}^{t-1}\{\phi_{s}=\phi_{s}^{\prime}\}\subseteq
\{u_{t}^{\phi}=u_{t}^{\phi^{\prime}}\}. \label{eq:dhg3sbdj}%
\end{equation}
The family of such strategies will be denoted by~$\Phi^{\mathrm{b}}(Y,X)$.
\end{definition}

The buyer's price of the game option in currency~$j$ can be understood as the
largest amount in that currency which can be raised against a long position in
the option used as surety. The precise definition is as follows.

\begin{definition}
\label{Def:nfn457w6dbg45d}\upshape The \emph{buyer's }(or \emph{bid})\emph{
price} in currency $j=1,\ldots,d$ of a game option $(Y,X)$ under gradual
exercise and cancellation is defined as%
\[
\pi_{j}^{\mathrm{b}}(Y,X):=\sup\left\{  -x\in\mathbb{R}\,|\,\exists(\psi
,u)\in\Phi^{\mathrm{b}}(Y,X):xe^{j}=u_{0}\right\}  .
\]

\end{definition}

\subsection{Buyer's pricing algorithm}

As is well known, there is a symmetry between the buyer's and seller's
superhedging and pricing problems for a European option. The symmetry
consists, essentially, in reversing the sign of the payoff while also
reversing the roles of buyer and seller. Hence, solving the seller's problem
also yields a solution to the buyer's problem, and \emph{vice versa}. However,
for an American option this symmetry is broken, and one needs to solve the
buyer's and seller's problems separately; see for example
\cite{rouxzastawniak2014} or \cite{rouxzastawniak2016}.

On first sight, it might appear that the symmetry between the buyer and seller
might be restored in the case of a game option. However, in fact, this is not
so when the buyer has priority to exercise the option before the seller can
cancel it. Reversing their roles would give priority to the seller. Combined
with condition~(\ref{eq:hhjfd6487shd6543}), this breaks the symmetry, and so a
specific solution to the buyer's problem is needed. This is facilitated by the
following construction.

\begin{construction}\label{constr:buyer}
\upshape Construct adapted sequences $\mathcal{Y}_{t}^{\mathrm{b}}%
,\mathcal{X}_{t}^{\mathrm{b}},\mathcal{V}_{t}^{\mathrm{b}},\mathcal{W}%
_{t}^{\mathrm{b}},\mathcal{Z}_{t}^{\mathrm{b}}$ for $t=0,\ldots,T$ as follows.
First, put%
\[
\mathcal{Y}_{t}^{\mathrm{b}}:=-Y_{t}+\mathcal{Q}_{t},\quad\mathcal{X}%
_{t}^{\mathrm{b}}:=-X_{t}+\mathcal{Q}_{t}%
\]
for all $t=0,\ldots,T$ and%
\[
\mathcal{W}_{T}^{\mathrm{b}}:=\mathcal{V}_{T}^{\mathrm{b}}:=\mathcal{L}%
_{T},\quad\mathcal{Z}_{T}^{\mathrm{b}}:=\mathcal{Y}_{T}^{\mathrm{b}}.
\]
Then, for $t=T-1,\ldots,0$ define by backward induction%
\begin{align*}
\mathcal{W}_{t}^{\mathrm{b}}  &  :=\mathcal{Z}_{t+1}^{\mathrm{b}}%
\cap\mathcal{L}_{t},\\
\mathcal{V}_{t}^{\mathrm{b}}  &  :=\mathcal{W}_{t}^{\mathrm{b}}+\mathcal{Q}%
_{t},\\
\mathcal{Z}_{t}^{\mathrm{b}}  &  :=\operatorname*{conv}\{\mathcal{V}%
_{t}^{\mathrm{b}}\cap\mathcal{X}_{t}^{\mathrm{b}},\mathcal{Y}_{t}^{\mathrm{b}%
}\}.
\end{align*}

\end{construction}

As compared to the seller's Construction~\ref{constr:seller}, apart from
swapping the payoff processes $Y,X$ for $-X,-Y$, which would have been
enough had there been a simple symmetry between the buyer and seller, the
operations of intersection and convex hull are taken in the reverse order in
the last line of this construction.

The proofs of the results below concerning the buyer's case resemble those for the seller, but certain details follow a diverse pattern to
account for the differences between the seller's and buyer's pricing constructions.

Just as in the seller's case, the same argument as in the proof of
Proposition~5.1 in~\cite{rouxzastawniak2014} shows that the~$\mathcal{Z}%
_{t}^{\mathrm{b}}$ are polyhedral convex sets. Moreover, we shall see that
$\mathcal{Z}_{0}^{\mathrm{b}}$ plays a similar role for the buyer as
$\mathcal{Z}_{0}^{\mathrm{a}}$ does for the seller, namely it is the set of all
initial endowments allowing the option buyer to superhedge their position.
This leads to the following result.

\begin{theorem}
\label{Thm:fjf75s41s}The buyer's price in currency $j=1,\ldots,d$ of a game
option~$(Y,X)$ with gradual exercise and cancellation can be expressed as%
\[
\pi_{j}^{\mathrm{b}}(Y,X)=\max\left\{  -x\in\mathbb{R}\,|\,xe^{j}%
\in\mathcal{Z}_{0}^{\mathrm{b}}\right\}  .
\]

\end{theorem}

To prove this theorem we need the following family $\Lambda^{\mathrm{b}}(Y,X)$, the elements of which can
be seen as strategies superhedging the buyer's position in a game option with
instant (rather than gradual) cancellation, gradual exercise and deferred
(rather than immediate) solvency.

\begin{definition}
\upshape We define $\Lambda^{\mathrm{b}}(Y,X)$ as the family consisting of all
pairs $(\psi,z)$, where $\psi\in\mathcal{X}$ and $z\in\Phi$, that satisfy the
conditions%
\begin{align*}
z_{t}+\psi_{t}Y_{t}-z_{t+1}  &  \in\mathcal{Q}_{t}\quad\text{for all
}t=0,\ldots,T-1,\\
z_{t}+\psi_{t}Y+\psi_{t+1}^{\ast}X_{t}  &  \in\mathcal{Q}_{t}\quad\text{for
all }t=0,\ldots,T.
\end{align*}

\end{definition}

The next two results are similar to Propositions~\ref{Prop:7djhb78a}
and~\ref{Prop:7d4hnk0a}. First,~$\mathcal{Z}_{0}^{\mathrm{b}}$ is shown to be
equal to the set of initial endowments for the strategies in $\Lambda
^{\mathrm{b}}(Y,X)$.

\begin{proposition}
\label{Prop:9h75hhsl}%
\[
\mathcal{Z}_{0}^{\mathrm{b}}=\left\{  z_{0}\in\mathbb{R}^{d}\,|\,(\psi
,z)\in\Lambda^{\mathrm{b}}(Y,X)\right\}  .
\]

\end{proposition}

The set of initial endowments for the strategies in $\Lambda^{\mathrm{b}%
}(Y,X)$ is then shown to coincide with that for the strategies in
$\Phi^{\mathrm{b}}(Y,X)$.

\begin{proposition}
\label{Prop:jg768enk}%
\[
\left\{  z_{0}\in\mathbb{R}^{d}\,|\,(\psi,z)\in\Lambda^{\mathrm{b}%
}(Y,X)\right\}  =\left\{  u_{0}\in\mathbb{R}^{d}\,|\,(\psi,u)\in
\Phi^{\mathrm{b}}(Y,X)\right\}  .
\]

\end{proposition}

The proofs of these two propositions are in the Appendix,
Section~\ref{Sect:Appendix}. Once these results have been established, proving
that~$\mathcal{Z}_{0}^{\mathrm{b}}$ is the set of initial endowments for the
strategies in $\Phi^{\mathrm{b}}(Y,X)$, Theorem~\ref{Thm:fjf75s41s} follows;
for details, see the proof in the Appendix, Section~\ref{Sect:Appendix}.

\subsection{Buyer's price representation}

In this section we obtain a representation of the buyer's price of a game
option with gradual exercise, by exploiting a link with the price of an
American option with gradual exercise and payoff process $-Q_{\,\cdot\,,\psi
}=(-Q_{t,\psi})_{t=0}^{T}$ defined for any $\psi\in\chi$, where%
\[
Q_{t,\psi}:=Q_{\chi^{t},\psi}\text{\quad for }t=0,\ldots,T.
\]
Such a link is furnished by the next lemma.

\begin{lemma}
\label{Lem:k856dg3a}For any $\psi\in\mathcal{X}$%
\[
\left\{  u_{0}\,|\,(\psi,u)\in\Phi^{\mathrm{b}}(Y,X)\right\}  =\left\{
z_{0}\,|\,z\in\Psi^{\mathrm{a}}(-Q_{\,\cdot\,,\psi})\right\}  ,
\]
where $-Q_{\,\cdot\,,\psi}=(-Q_{t,\psi})_{t=0}^{T}$ is the payoff process of
an American option.
\end{lemma}

With the aid of this lemma, in a similar manner as in
Section~\ref{Sect:sel_price_rep}, we can establish the following
representation of the buyer's price.

\begin{theorem}
\label{Thm:dhhd5436syy}The buyer's price in currency $j=1,\ldots,d$ of a game
option~$(Y,X)$ with gradual exercise and cancellation can be represented as%
\[
\pi_{j}^{\mathrm{b}}(Y,X)=\max_{\psi\in\mathcal{X}}\min_{\phi\in\mathcal{X}%
}\min_{(\mathbb{Q},S)\in\mathcal{\bar{P}}_{j}^{\mathrm{d}}(\phi)}%
\mathbb{E}_{\mathbb{Q}}(\left(  Q_{\,\cdot\,,\psi}\cdot S\right)  _{\phi}).
\]

\end{theorem}

The proofs of Lemma~\ref{Lem:k856dg3a} and Theorem~\ref{Thm:dhhd5436syy} can
be found in the Appendix, Section~\ref{Sect:Appendix}.

\section{Example}\label{sec:6}

A game option $(Y,X)$ in a binary two-step two-currency model is presented in Figure~\ref{fig:two-step-model}. The model is recombinant; the option payoff is path-independent and has no cancellation penalties at time $2$. The model has transaction costs only at the node $\mathrm{u}$ at time $1$.
  \begin{figure}
 \begin{center}
 \begin{tikzpicture}[model,twostep, max label lines = 2]
  \node (root) {$\begin{gathered}\tfrac{1}{\pi^{12}_0} = \pi^{21}_0 = 10\\
	\begin{aligned}Y_0&=\twovector{0}{0} \\ X_0&=\twovector{0}{5}\end{aligned}\end{gathered}$}
    child { node[label=left:u] (u) {$\begin{gathered}\begin{aligned}\pi^{12}_1 &=\tfrac{1}{8} \\ \pi^{21}_1 &= 16\end{aligned}\\
	\begin{aligned}Y_1&=\twovector{0}{3} \\ X_1&=\twovector{0}{6} \end{aligned}\end{gathered}$}
    	child { node[label=right:uu] (uu) {$\begin{gathered}\tfrac{1}{\pi^{12}_2} =\pi^{21}_2 = 14\\
	Y_2=X_2=\twovector{0}{9}\end{gathered}$}}
    	child { node [draw=none] {} edge from parent [draw=none]}
	}
    child { node[label=left:d] (d) {$\begin{gathered}\tfrac{1}{\pi^{12}_1} =\pi^{21}_1 = 6\\
	\begin{aligned}Y_1&=\twovector{0}{0} \\ X_1&=\twovector{0}{3}\end{aligned}\end{gathered}$}
    	child { node [draw=none] {} edge from parent [draw=none]}
    	child { node[label=right:dd] (dd) {$\begin{gathered}\tfrac{1}{\pi^{12}_2} = \pi^{21}_2 = 4\\
	Y_2=X_2=\twovector{0}{0}\end{gathered}$}}
	};

	
	\node[label=right:{ud,du}] (ud) at (2\leveldistance,0) {$\begin{gathered}\tfrac{1}{\pi^{12}_2} = \pi^{21}_2 = 10\\
	Y_2=X_2=\twovector{0}{4}\end{gathered}$};
	
    \draw (u) -- (ud);
    \draw (d) -- (ud);
\end{tikzpicture}
\end{center}
  \caption{Game option in binary two-step two-currency model}
  \label{fig:two-step-model}
 \end{figure}
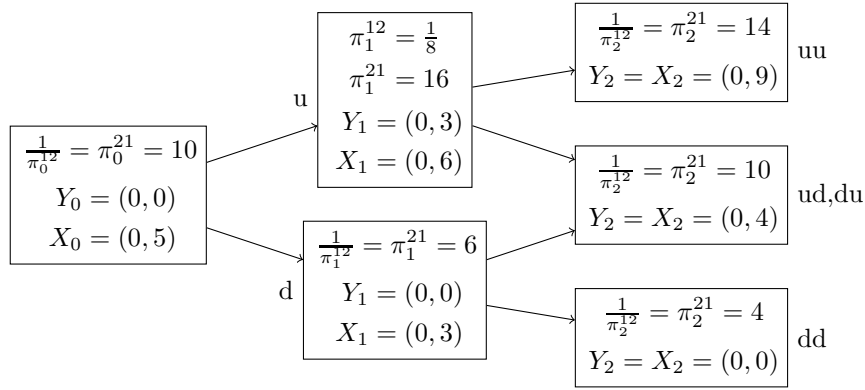

   Constructions 3.1 and 3.4 of \cite{roux2016} give the bid-ask spread of the game option $(Y,X)$ with instant exercise and cancellation in terms of currency $2$ to be $[3.2,5]$. We will show below that the bid-ask spread of $(Y,X)$ with gradual exercise and cancellation is
  \[
  [\pi_2^{\mathrm{b}}(Y,X), \pi_2^{\mathrm{a}}(Y,X)] = [\tfrac{11}{3},\tfrac{14}{3}] \approx [3.6667,4.6667] \subset [3.2,5].
 \]
  (Indeed the bid and ask prices of $(Y,X)$ can be read off the vertical axes in Figures~\ref{fig:ex:toy:Za0} and~\ref{fig:ex:toy:Zb0} below.) Thus gradual exercise and cancellation leads to a smaller bid-ask spread in this example.

 Let us use Construction~\ref{constr:seller} to find the set $\mathcal{Z}^{\mathrm{a}}_0$ of initial endowments that allow the seller to superhedge $(Y,X)$ with gradual exercise and cancellation. At time $t=2$ we have
 \begin{align*}
  \mathcal{Z}^{\mathrm{auu}}_2 &= \{(x^1,x^2)\in\mathbb{R}^2:14x^1+x^2 \ge 9\}, \\
  \mathcal{Z}^{\mathrm{aud}}_2 &= \mathcal{Z}^{\mathrm{adu}}_2 =  \{(x^1,x^2)\in\mathbb{R}^2:10x^1+x^2 \ge 4\}, \\
  \mathcal{Z}^{\mathrm{add}}_2 &= \{(x^1,x^2)\in\mathbb{R}^2:4x^1+x^2 \ge 0\}.
 \end{align*}
 Figure~\ref{fig:ex:toy:Zau1} illustrates the construction at time $t=1$ at the node $\mathrm{u}$, which results in
 \[
  \mathcal{Z}^{\mathrm{au}}_1 = \{(x^1,x^2)\in\mathbb{R}^2:14x^1+x^2 \ge 6, \tfrac{58}{5}x^1+x^2 \ge 6, 10x^1+x^2 \ge 4\}.
 \]
 Similar considerations at the node $\mathrm{d}$ give
 \[
  \mathcal{Z}^{\mathrm{ad}}_1 = \{(x^1,x^2)\in\mathbb{R}^2:6x^1+x^2 \ge \tfrac{4}{3}\}.
 \]
 The construction at time $t=0$ gives
 \[
  \mathcal{Z}^{\mathrm{a}}_0 = \{(x^1,x^2)\in\mathbb{R}^2:10x^1+x^2 \ge \tfrac{14}{3}\},
 \]
 as illustrated in Figure~\ref{fig:ex:toy:Za0}.
    \begin{figure}
  \newcommand{\minx}{-0.3}
  \newcommand{\maxx}{1.5}
  \newcommand{\miny}{-10}
  \newcommand{\maxy}{10}
   \centering
   \begin{tabular}{@{}rr@{}}
     \footnotesize\begin{tikzpicture}[x=0.5/(\maxx-(\minx))*\figurewidth,y=0.5/(\maxy-(\miny))*\figureheight]
    \draw[draw=none,pattern=north east lines, pattern color=lightgray] ({0 - (\maxy - 9)/14} ,\maxy) -- (0, 9) -- ({0-(\miny-9)/14},\miny) -- (\maxx,\miny) -- (\maxx,\maxy) -- (\minx,\maxy) -- cycle; 
     \draw[draw=none,pattern=north west lines, pattern color=lightgray] (\minx ,{4-(\minx-0)*10}) -- (0, 4) -- ({0-(\miny-4)/10},\miny) -- (\maxx,\miny) -- (\maxx,\maxy) -- (\minx,\maxy) -- cycle; 
   \draw[draw=none,fill=lightgray,opacity=0.5] ({0 - (\maxy - 9)/14} ,\maxy) -- (5/4, -17/2) -- ({0-(\miny-4)/10},\miny) -- (\maxx,\miny) -- (\maxx,\maxy) -- cycle; 
      \foreach \y/\yname in {{-17/2}/{-\tfrac{17}{2}},9,4}
         \draw[tick] (0,\y) node[left,ticklabel] {\footnotesize$\yname$} -- ++(1ex,0);

       \foreach \x/\xname in {{5/4}/\tfrac{5}{4},{2/5}/\tfrac{2}{5}}
         \draw[tick] (\x,0) node[below,ticklabel] {\footnotesize$\xname$} -- ++(0,1ex);

       \foreach \x/\xname in {{9/14}/\tfrac{9}{14}}
         \draw[tick] (\x,0) -- ++(0,1ex) node[above,ticklabel] {\footnotesize$\xname$};

    \draw[axis] (\minx,0) -- (\maxx,0) node[above left,axislabel] {$x^1$};
    \draw[axis] (0,\maxy) -- (0,\miny) node[above right,axislabel] {$x^2$};
    \draw[thick] ({0 - (\maxy - 9)/14} ,\maxy) -- (5/4, -17/2) -- ({0-(\miny-4)/10},\miny); 
    \draw (\maxx,\maxy) node[below left] {$\mathcal{W}^{\mathrm{au}}_1=\mathcal{Z}^{\mathrm{auu}}_2\cap\mathcal{Z}^{\mathrm{aud}}_2$};
     \draw (5/4,-17/2) node {$\bullet$};
     \draw (5/4,-17/2) node[above right] {$z^{\mathrm{au}}_2$};
     \draw ({0 - (\maxy - 9)/14} ,\maxy) -- (0, 9) -- ({0-(\miny-9)/14},\miny) node[sloped,above,pos=0.2] {$\mathcal{Z}^{\mathrm{auu}}_2$};
     \draw (\minx ,{4-(\minx-0)*10}) -- (0, 4) -- ({0-(\miny-4)/10},\miny) node[sloped,above,pos=0.12] {$\mathcal{Z}^{\mathrm{aud}}_2$};
  \end{tikzpicture}
 &
     \footnotesize\begin{tikzpicture}[x=0.5/(\maxx-(\minx))*\figurewidth,y=0.5/(\maxy-(\miny))*\figureheight]
   \draw[draw=none,pattern=north east lines, pattern color=lightgray] ({0 - (\maxy - 9)/14} ,\maxy) -- (5/4, -17/2) -- ({0-(\miny-4)/10},\miny) -- (\maxx,\miny) -- (\maxx,\maxy) -- cycle; 
   \draw[draw=none,fill=lightgray,opacity=0.5] ({0 - (\maxy - 9)/14} ,\maxy) -- (5/4, -17/2) -- ({0-(\miny-4)/10},\miny) -- (\maxx,\miny) -- (\maxx,\maxy) -- cycle; 
%
      \foreach \y/\yname in {{-17/2}/{-\tfrac{17}{2}},9}
         \draw[tick] (0,\y) node[left,ticklabel] {\footnotesize$\yname$} -- ++(1ex,0);

       \foreach \x/\xname in {{5/4}/\tfrac{5}{4},{9/14}/\tfrac{9}{14}}
         \draw[tick] (\x,0) node[below,ticklabel] {\footnotesize$\xname$} -- ++(0,1ex);
     \draw[thick] ({0 - (\maxy - 9)/14} ,\maxy) -- (5/4, -17/2) node[sloped,above,pos=0.5] {$\mathcal{W}^{\mathrm{au}}_1$} -- ({0-(\miny-4)/10},\miny);
    \draw (\maxx,\maxy) node[below left] {$\mathcal{V}^{\mathrm{au}}_1=\mathcal{W}^{\mathrm{au}}_1+ \mathcal{Q}^\mathrm{u}_1$};
    \draw[dashed] (\minx, {0 - (\minx-0)*14}) -- (0, 0) -- ({0 - (\miny-0)/10},\miny)  node[sloped, above, pos=0.6] {lower boundary of $\mathcal{Q}^\mathrm{u}_1$};
     \draw (5/4,-17/2) node {$\bullet$};
     \draw (5/4,-17/2) node[above right] {$z^{\mathrm{au}}_2$};
     \draw[axis] (\minx,0) -- (\maxx,0) node[above left,axislabel] {$x^1$};
     \draw[axis] (0,\maxy) -- (0,\miny) node[above right,axislabel] {$x^2$};
  \end{tikzpicture}
 \\
     \footnotesize\begin{tikzpicture}[x=0.5/(\maxx-(\minx))*\figurewidth,y=0.5/(\maxy-(\miny))*\figureheight]
    \draw[draw=none,pattern=north east lines, pattern color=lightgray]  ({0 - (\maxy - 9)/14} ,\maxy) -- (5/4, -17/2) -- ({0-(\miny-4)/10},\miny) -- (\maxx,\miny) -- (\maxx,\maxy) -- cycle; 
    \draw[draw=none,pattern=north west lines, pattern color=lightgray]  ({0 - (\maxy-6)/14},\maxy) -- (0, 6) -- (\maxx, {6 - (\maxx-0)*10}) -- (\maxx,\maxy) -- (\minx,\maxy) -- cycle; 
    \draw[draw=none,fill=lightgray,opacity=0.5] ({0 - (\maxy-6)/14},\maxy) -- (0, 6) -- (5/4, -17/2) -- ({0-(\miny-4)/10},\miny) -- (\maxx,\miny) -- (\maxx,\maxy) -- cycle; 
      \foreach \y/\yname in {{-17/2}/{-\tfrac{17}{2}},{-11/3}/{-\tfrac{11}{3}},6,9}
         \draw[tick] (0,\y) node[left,ticklabel] {\footnotesize$\yname$} -- ++(1ex,0);
       \foreach \x/\xname in {{15/29}/\tfrac{15}{29},{5/4}/\tfrac{5}{4}}
         \draw[tick] (\x,0) node[below,ticklabel] {\footnotesize$\xname$} -- ++(0,1ex);
       \foreach \x/\xname in {{9/14}/\tfrac{9}{14},{5/6}/\tfrac{5}{6}}
         \draw[tick] (\x,0) -- ++(0,1ex) node[above,ticklabel] {\footnotesize$\xname$};
     \draw[axis] (\minx,0) -- (\maxx,0) node[above left,axislabel] {$x^1$};
     \draw[axis] (0,\maxy) -- (0,\miny) node[above right,axislabel] {$x^2$};
     \draw (5/6,-11/3) node {$\bullet$};
     \draw (5/6,-11/3) node[below left] {$z^{\mathrm{a}}_1$};
     \draw (5/4,-17/2) node {$\bullet$};
     \draw (5/4,-17/2) node[right] {$z^{\mathrm{au}}_2$};
     \draw ({0 - (\maxy - 9)/14} ,\maxy) -- (5/4, -17/2) node[sloped,above,pos=0.3] {$\mathcal{V}^{\mathrm{au}}_1$} -- ({0-(\miny-4)/10},\miny);
    \draw ({0 - (\maxy-6)/14},\maxy) -- (0, 6) -- (\maxx, {6 - (\maxx-0)*10})  node[sloped, above, pos=0.8] {$\mathcal{X}^{\mathrm{au}}_1$};

        \draw[thick] ({0 - (\maxy-6)/14},\maxy) -- (0, 6) -- (5/4, -17/2) -- ({0-(\miny-4)/10},\miny);
     \draw (\maxx,\maxy) node[below left] {$\conv\{\mathcal{V}^{\mathrm{au}}_1,\mathcal{X}^{\mathrm{au}}_1\}$};
  \end{tikzpicture}
 &
    \footnotesize\begin{tikzpicture}[x=0.5/(\maxx-(\minx))*\figurewidth,y=0.5/(\maxy-(\miny))*\figureheight]
    \draw[draw=none,pattern=north east lines, pattern color=lightgray]  ({0 - (\maxy-6)/14},\maxy) -- (0, 6) -- (5/4, -17/2) -- ({0-(\miny-4)/10},\miny) -- (\maxx,\miny) -- (\maxx,\maxy) -- cycle; 
    \draw[draw=none,pattern=north west lines, pattern color=lightgray]  (\minx, {3 - (\minx-0)*14}) -- (0, 3) -- ({0 - (\miny-3)/10},\miny) -- (\maxx,\miny) -- (\maxx,\maxy) -- (\minx,\maxy) -- cycle; 
    \draw[draw=none,fill=lightgray,opacity=0.5] ({0 - (\maxy-6)/14},\maxy) -- (0, 6) -- (5/4, -17/2) -- ({0-(\miny-4)/10},\miny) -- (\maxx,\miny) -- (\maxx,\maxy) -- cycle; 
%
%
      \foreach \y/\yname in {{-17/2}/{-\tfrac{17}{2}},{-11/3}/{-\tfrac{11}{3}},6,3}
         \draw[tick] (0,\y) node[left,ticklabel] {\footnotesize$\yname$} -- ++(1ex,0);
       \foreach \x/\xname in {{5/4}/\tfrac{5}{4},{3/10}/\tfrac{3}{10}}
         \draw[tick] (\x,0) node[below,ticklabel] {\footnotesize$\xname$} -- ++(0,1ex);
       \foreach \x/\xname in {{15/29}/\tfrac{15}{29},{5/6}/\tfrac{5}{6}}
         \draw[tick] (\x,0) -- ++(0,1ex) node[above,ticklabel] {\footnotesize$\xname$};
     \draw (5/6,-11/3) node {$\bullet$};
     \draw (5/6,-11/3) node[above] {$z^{\mathrm{a}}_1$};
     \draw[axis] (\minx,0) -- (\maxx,0) node[above left,axislabel] {$x^1$};
     \draw[axis] (0,\maxy) -- (0,\miny) node[above right,axislabel] {$x^2$};
    \draw[thick] ({0 - (\maxy-6)/14},\maxy) -- (0, 6) -- (5/4, -17/2)  node[sloped, above, pos=0.1] {$\conv\{\mathcal{V}^{\mathrm{au}}_1,\mathcal{X}^{\mathrm{au}}_1\}$} -- ({0-(\miny-4)/10},\miny);
    \draw  (\minx, {3 - (\minx-0)*14}) -- (0, 3) -- ({0 - (\miny-3)/10},\miny)  node[sloped, below, pos=0.7] {$\mathcal{Y}^{\mathrm{au}}_1$};
 \draw (\maxx,\maxy) node[below left] {$\mathcal{Z}^{\mathrm{au}}_1=\conv\{\mathcal{V}^{\mathrm{au}}_1,\mathcal{X}^{\mathrm{au}}_1\cap\mathcal{Y}^{\mathrm{au}}_1\}$};
  \end{tikzpicture}
   \end{tabular}
    \caption{$\mathcal{W}^{\mathrm{au}}_1$, $\mathcal{V}^{\mathrm{au}}_1$, $\conv\{\mathcal{V}^{\mathrm{au}}_1,\mathcal{X}^{\mathrm{au}}_1\}$, $\mathcal{Z}^{\mathrm{au}}_1$, $z^{\mathrm{a}}_1$, $z^{\mathrm{au}}_2$}
    \label{fig:ex:toy:Zau1}
   \end{figure}
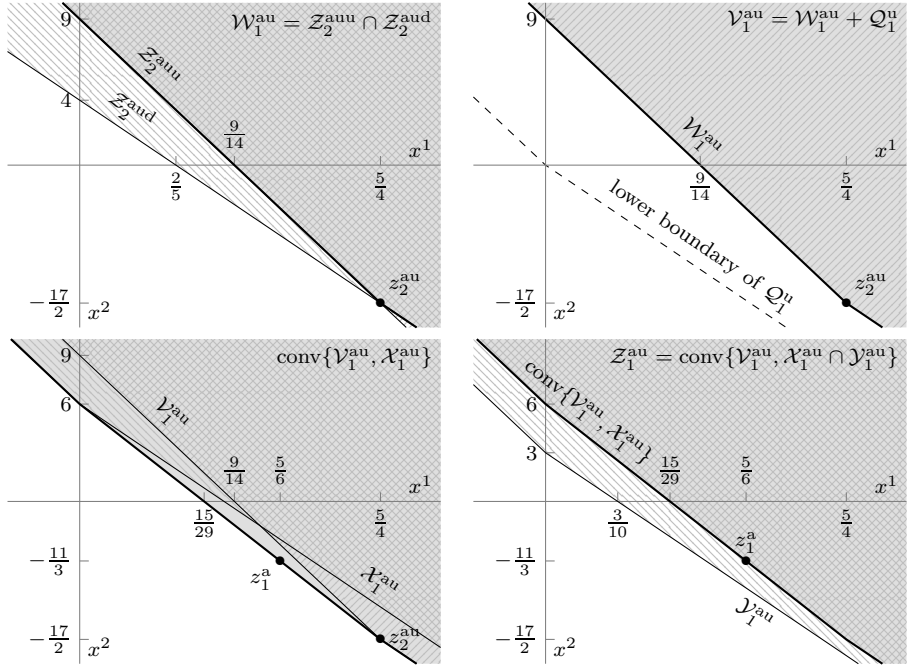

   A superhedging strategy for the seller starting from the initial endowment $\twovector{0}{\pi_2^{\mathrm{a}}(Y,X)}$ can be constructed by following similar lines as in the proof of Proposition~\ref{Prop:7djhb78a} to assemble $(\phi,z^{\mathrm{a}})\in\Lambda^\mathrm{a}(Y,X)$, and then converting it into a superhedging strategy $(\phi,u^{\mathrm{a}})\in\Phi^\mathrm{a}(Y,X)$ using the arguments in the proof of Proposition~\ref{Prop:7d4hnk0a}. We illustrate the first part of the process here for the scenario~$\mathrm{uu}$. Define first $z^{\mathrm{a}}_0:=\twovector{0}{\tfrac{14}{3}}$; it is clear from Figure~\ref{fig:ex:toy:Za0} that $z^{\mathrm{a}}_0\notin\mathcal{X}^{\mathrm{a}}_0$, leading to $\phi_0:=0$. Choosing $z^{\mathrm{a}}_1:=\twovector{\tfrac{5}{6}}{-\tfrac{11}{3}}\in\mathcal{W}^{\mathrm{a}}_0$ then gives that
  \[
   z^{\mathrm{a}}_0-\phi_0X_0 - z^{\mathrm{a}}_1 = z^{\mathrm{a}}_0-z^{\mathrm{a}}_1\in\mathcal{Q}_0.
  \]
  Figure~\ref{fig:ex:toy:Zau1} shows that $z^{\mathrm{a}}_1\in\mathcal{Z}^{\mathrm{a}}_1\subset\mathcal{Y}^{\mathrm{a}}_1$. It also shows that defining $\phi_1^\mathrm{u} := \tfrac{2}{3}$ and $z^{\mathrm{au}}_2 := \twovector{\tfrac{5}{4}}{-\tfrac{17}{2}}$ leads to
  \[
   z^{\mathrm{a}}_1 = \phi_1^\mathrm{u}\twovector{0}{6} + (1-\phi_1^\mathrm{u})z^{\mathrm{au}}_2,
  \]
  with $\twovector{0}{6}\in\mathcal{X}^{\mathrm{au}}_1$ and $z^{\mathrm{au}}_2\in\mathcal{V}^{\mathrm{au}}_1=\mathcal{W}^{\mathrm{au}}_1\subset\mathcal{Z}^{\mathrm{auu}}_2=\mathcal{Y}^{\mathrm{auu}}_2$. Thus this strategy corresponds to cancellation of $\tfrac{1}{3}$ of the option at time $1$ at the node $\mathrm{u}$ and the remaining $\phi_1^\mathrm{uu} := \tfrac{1}{3}$ at time $2$.
   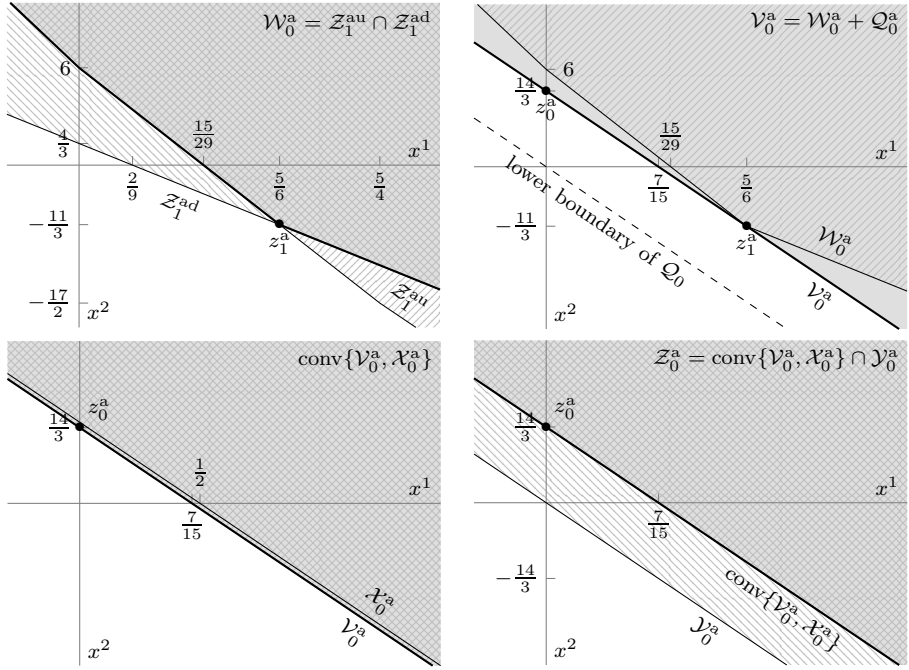
\begin{figure}
  \newcommand{\minx}{-0.3}
  \newcommand{\maxx}{1.5}
  \newcommand{\miny}{-10}
  \newcommand{\maxy}{10}
   \centering
   \begin{tabular}{@{}rr@{}}
     \footnotesize\begin{tikzpicture}[x=0.5/(\maxx-(\minx))*\figurewidth,y=0.5/(\maxy-(\miny))*\figureheight]
    \draw[draw=none,pattern=north east lines, pattern color=lightgray] ({0 - (\maxy-6)/14},\maxy) -- (0, 6) -- (5/4, -17/2) -- ({0-(\miny-4)/10},\miny) -- (\maxx,\miny) -- (\maxx,\maxy) -- (\minx,\maxy) -- cycle; 
     \draw[draw=none,pattern=north west lines, pattern color=lightgray] (\minx ,{4/3-(\minx-0)*6}) -- (0, 4/3) -- (\maxx ,{4/3-(\maxx-0)*6}) -- (\maxx,\maxy) -- (\minx,\maxy) -- cycle; 
   \draw[draw=none,fill=lightgray,opacity=0.5] ({0 - (\maxy-6)/14},\maxy) -- (0, 6) -- (5/6, -11/3) -- (\maxx ,{4/3-(\maxx-0)*6}) -- (\maxx,\maxy) -- cycle; 
       \foreach \x/\xname in {{15/29}/\tfrac{15}{29}}
         \draw[tick] (\x,0) -- ++(0,1ex) node[above,ticklabel] {\footnotesize$\xname$};

         \foreach \y/\yname in {{-17/2}/{-\tfrac{17}{2}},{-11/3}/{-\tfrac{11}{3}},6,{4/3}/\tfrac{4}{3}}
         \draw[tick] (0,\y) node[left,ticklabel] {\footnotesize$\yname$} -- ++(1ex,0);

       \foreach \x/\xname in {{2/9}/\tfrac{2}{9},{5/4}/\tfrac{5}{4},{5/6}/\tfrac{5}{6}}
         \draw[tick] (\x,0) node[below,ticklabel] {\footnotesize$\xname$} -- ++(0,1ex);

    \draw[axis] (\minx,0) -- (\maxx,0) node[above left,axislabel] {$x^1$};
    \draw[axis] (0,\maxy) -- (0,\miny) node[above right,axislabel] {$x^2$};
     \draw (5/6,-11/3) node {$\bullet$};
     \draw (5/6,-11/3) node[below] {$z^{\mathrm{a}}_1$};
    \draw[thick] ({0 - (\maxy-6)/14},\maxy) -- (0, 6) -- (5/6, -11/3) -- (\maxx ,{4/3-(\maxx-0)*6}); 
    \draw (\maxx,\maxy) node[below left] {$\mathcal{W}^{\mathrm{a}}_0=\mathcal{Z}^{\mathrm{au}}_1\cap\mathcal{Z}^{\mathrm{ad}}_1$};
    \draw ({0 - (\maxy-6)/14},\maxy) -- (0, 6) -- (5/4, -17/2) -- ({0-(\miny-4)/10},\miny)  node[sloped, above, midway] {$\mathcal{Z}^{\mathrm{au}}_1$};
     \draw (\minx ,{4/3-(\minx-0)*6}) -- (0, 4/3) -- (\maxx ,{4/3-(\maxx-0)*6}) node[sloped,below,pos=0.3] {$\mathcal{Z}^{\mathrm{ad}}_1$};
  \end{tikzpicture}
 &
     \footnotesize\begin{tikzpicture}[x=0.5/(\maxx-(\minx))*\figurewidth,y=0.5/(\maxy-(\miny))*\figureheight]
    \draw[draw=none,pattern=north east lines, pattern color=lightgray] ({0 - (\maxy-6)/14},\maxy) -- (0, 6) -- (5/6, -11/3) -- (\maxx ,{4/3-(\maxx-0)*6}) -- (\maxx,\maxy) -- cycle; 
   \draw[draw=none,fill=lightgray,opacity=0.5] (\minx, {14/3 - (\minx-0)*10}) -- (0, 14/3) -- ({0 - (\miny-14/3)/10},\miny) -- (\maxx,\miny) -- (\maxx,\maxy) -- (\minx,\maxy) -- cycle; 
      \foreach \y/\yname in {{-11/3}/{-\tfrac{11}{3}},{14/3}/{\tfrac{14}{3}}}
         \draw[tick] (0,\y) node[left,ticklabel] {\footnotesize$\yname$} -- ++(1ex,0);

      \foreach \y/\yname in {6}
         \draw[tick] (0,\y) -- ++(1ex,0) node[right,ticklabel] {\footnotesize$\yname$};

       \foreach \x/\xname in {{15/29}/\tfrac{15}{29}}
         \draw[tick] (\x,0) -- ++(0,1ex) node[above,ticklabel] {\footnotesize$\xname$};

       \foreach \x/\xname in {{7/15}/\tfrac{7}{15},{5/6}/\tfrac{5}{6}}
         \draw[tick] (\x,0) node[below,ticklabel] {\footnotesize$\xname$} -- ++(0,1ex);
     \draw[axis] (\minx,0) -- (\maxx,0) node[above left,axislabel] {$x^1$};
     \draw[axis] (0,\maxy) -- (0,\miny) node[above right,axislabel] {$x^2$};
     \draw (0,14/3) node {$\bullet$};
     \draw (0,14/3) node[below] {$z^{\mathrm{a}}_0$};

     \draw (5/6,-11/3) node {$\bullet$};
     \draw (5/6,-11/3) node[below] {$z^{\mathrm{a}}_1$};
     \draw ({0 - (\maxy-6)/14},\maxy) -- (0, 6) -- (5/6, -11/3) -- (\maxx ,{4/3-(\maxx-0)*6}) node[sloped,above,pos=0.5] {$\mathcal{W}^{\mathrm{a}}_0$};
    \draw[thick] (\minx, {14/3 - (\minx-0)*10}) -- (0, 14/3) -- ({0 - (\miny-14/3)/10},\miny)  node[sloped, below, pos=0.8] {$\mathcal{V}^{\mathrm{a}}_0$};
     \draw (\maxx,\maxy) node[below left] {$\mathcal{V}^{\mathrm{a}}_0=\mathcal{W}^{\mathrm{a}}_0+\mathcal{Q}^{\mathrm{a}}_0$};
    \draw[dashed] (\minx, {0 - (\minx-0)*10}) -- (0, 0)-- ({0 - (\miny-0)/10},\miny)   node[sloped, below, near start] {lower boundary of $\mathcal{Q}_0$};
  \end{tikzpicture}
 \\
     \footnotesize\begin{tikzpicture}[x=0.5/(\maxx-(\minx))*\figurewidth,y=0.5/(\maxy-(\miny))*\figureheight]
    \draw[draw=none,pattern=north east lines, pattern color=lightgray] (\minx, {14/3 - (\minx-0)*10}) -- (0, 14/3) -- ({0 - (\miny-14/3)/10},\miny) -- (\maxx,\miny) -- (\maxx,\maxy) -- (\minx,\maxy) -- cycle; 
    \draw[draw=none,pattern=north west lines, pattern color=lightgray]  (\minx, {5 - (\minx-0)*10}) -- (0, 5) -- ({0 - (\miny-5)/10},\miny) -- (\maxx,\miny) -- (\maxx,\maxy) -- (\minx,\maxy) -- cycle; 
    \draw[draw=none,fill=lightgray,opacity=0.5] (\minx, {14/3 - (\minx-0)*10}) -- (0, 14/3) -- ({0 - (\miny-14/3)/10},\miny) -- (\maxx,\miny) -- (\maxx,\maxy) -- (\minx,\maxy) -- cycle; 
      \foreach \y/\yname in {{14/3}/{\tfrac{14}{3}}}
         \draw[tick] (0,\y) node[left,ticklabel] {\footnotesize$\yname$} -- ++(1ex,0);

       \foreach \x/\xname in {{1/2}/\tfrac{1}{2}}
         \draw[tick] (\x,0) -- ++(0,1ex) node[above,ticklabel] {\footnotesize$\xname$};

         \foreach \x/\xname in {{7/15}/\tfrac{7}{15}}
         \draw[tick] (\x,0) node[below,ticklabel] {\footnotesize$\xname$} -- ++(0,1ex);
     \draw[axis] (\minx,0) -- (\maxx,0) node[above left,axislabel] {$x^1$};
     \draw[axis] (0,\maxy) -- (0,\miny) node[above right,axislabel] {$x^2$};

     \draw (0,14/3) node {$\bullet$};
     \draw (0,14/3) node[above right] {$z^{\mathrm{a}}_0$};
\draw (\minx, {14/3 - (\minx-0)*10}) -- (0, 14/3) -- ({0 - (\miny-14/3)/10},\miny)  node[sloped, below, pos=0.8] {$\mathcal{V}^{\mathrm{a}}_0$};
     \draw (\minx, {5 - (\minx-0)*10}) -- (0, 5) -- ({0 - (\miny-5)/10},\miny)  node[sloped, above, pos=0.8] {$\mathcal{X}^{\mathrm{a}}_0$};

        \draw[thick] (\minx, {14/3 - (\minx-0)*10}) -- (0, 14/3) -- ({0 - (\miny-14/3)/10},\miny);
     \draw (\maxx,\maxy) node[below left] {$\conv\{\mathcal{V}^{\mathrm{a}}_0,\mathcal{X}^{\mathrm{a}}_0\}$};

  \end{tikzpicture}
 &
    \footnotesize\begin{tikzpicture}[x=0.5/(\maxx-(\minx))*\figurewidth,y=0.5/(\maxy-(\miny))*\figureheight]
    \draw[draw=none,pattern=north east lines, pattern color=lightgray]  (\minx, {14/3 - (\minx-0)*10}) -- (0, 14/3) -- ({0 - (\miny-14/3)/10},\miny) -- (\maxx,\miny) -- (\maxx,\maxy) -- (\minx,\maxy) -- cycle; 
    \draw[draw=none,pattern=north west lines, pattern color=lightgray]  (\minx, {0 - (\minx-0)*10}) -- (0, 0) -- ({0 - (\miny-0)/10},\miny) -- (\maxx,\miny) -- (\maxx,\maxy) -- (\minx,\maxy) -- cycle; 
    \draw[draw=none,fill=lightgray,opacity=0.5] (\minx, {14/3 - (\minx-0)*10}) -- (0, 14/3) -- ({0 - (\miny-14/3)/10},\miny) -- (\maxx,\miny) -- (\maxx,\maxy) -- (\minx,\maxy) -- cycle; 
      \foreach \y/\yname in {{14/3}/{\tfrac{14}{3}},{-14/3}/{-\tfrac{14}{3}}}
         \draw[tick] (0,\y) node[left,ticklabel] {\footnotesize$\yname$} -- ++(1ex,0);
         \foreach \x/\xname in {{7/15}/\tfrac{7}{15}}
         \draw[tick] (\x,0) node[below,ticklabel] {\footnotesize$\xname$} -- ++(0,1ex);
     \draw[axis] (\minx,0) -- (\maxx,0) node[above left,axislabel] {$x^1$};
     \draw[axis] (0,\maxy) -- (0,\miny) node[above right,axislabel] {$x^2$};

     \draw (0,14/3) node {$\bullet$};
     \draw (0,14/3) node[above right] {$z^{\mathrm{a}}_0$};
    \draw (\minx, {14/3 - (\minx-0)*10}) -- (0, 14/3) -- ({0 - (\miny-14/3)/10},\miny)  node[sloped, below, pos=0.7] {$\conv\{\mathcal{V}^{\mathrm{a}}_0,\mathcal{X}^{\mathrm{a}}_0\}$};
    \draw  (\minx, {0 - (\minx-0)*10}) -- (0, 0) -- ({0 - (\miny-0)/10},\miny)  node[sloped, below, pos=0.7] {$\mathcal{Y}^{\mathrm{a}}_0$};
    \draw[thick] (\minx, {14/3 - (\minx-0)*10}) -- (0, 14/3) -- ({0 - (\miny-14/3)/10},\miny);
 \draw (\maxx,\maxy) node[below left] {$\mathcal{Z}^{\mathrm{a}}_0=\conv\{\mathcal{V}^{\mathrm{a}}_0,\mathcal{X}^{\mathrm{a}}_0\}\cap\mathcal{Y}^{\mathrm{a}}_0$};
  \end{tikzpicture}
   \end{tabular}
    \caption{$\mathcal{W}^{\mathrm{a}}_0$, $\mathcal{V}^{\mathrm{a}}_0$, $\conv\{\mathcal{V}^{\mathrm{a}}_0,\mathcal{X}^{\mathrm{a}}_0\}$, $\mathcal{Z}^{\mathrm{a}}_0$, $z^{\mathrm{a}}_0$, $z^{\mathrm{a}}_1$}
    \label{fig:ex:toy:Za0}
   \end{figure}

     The set of superhedging strategies for the buyer of $(Y,X)$ can be computed by following Construction~\ref{constr:buyer}. At time $t=2$,
 \begin{align*}
  \mathcal{Z}^{\mathrm{buu}}_2 &= \{(x^1,x^2)\in\mathbb{R}^2:14x^1+x^2 \ge -9\}, \\
  \mathcal{Z}^{\mathrm{bud}}_2 &= \mathcal{Z}^{\mathrm{bdu}}_2 =  \{(x^1,x^2)\in\mathbb{R}^2:10x^1+x^2 \ge -4\}, \\
  \mathcal{Z}^{\mathrm{bdd}}_2 &= \{(x^1,x^2)\in\mathbb{R}^2:4x^1+x^2 \ge 0\}.
 \end{align*}
 The construction at time $t=1$ at the node $\mathrm{u}$ gives
 \[
  \mathcal{Z}^{\mathrm{bu}}_1 = \{(x^1,x^2)\in\mathbb{R}^2:14x^1+x^2 \ge -6, 10x^1+x^2 \ge -4\},
 \]
 as illustrated in Figure~\ref{fig:ex:toy:Zbu1}. Similar considerations at the node $\mathrm{d}$ result in
 \[
  \mathcal{Z}^{\mathrm{bd}}_1 = \{(x^1,x^2)\in\mathbb{R}^2:6x^1+x^2 \ge -\tfrac{4}{3}\}.
 \]
 Figure~\ref{fig:ex:toy:Zb0} demonstrates the construction at time $t=0$, which leads to
 \[
  \mathcal{Z}^{\mathrm{a}}_0 = \{(x^1,x^2)\in\mathbb{R}^2:10x^1+x^2 \ge -\tfrac{11}{3}\}.
 \]
 \begin{figure}
  \newcommand{\minx}{-1.5}
  \newcommand{\maxx}{0.3}
  \newcommand{\miny}{-10}
  \newcommand{\maxy}{12}
   \centering
   \begin{tabular}{@{}rr@{}}
     \footnotesize\begin{tikzpicture}[x=0.5/(\maxx-(\minx))*\figurewidth,y=0.5/(\maxy-(\miny))*\figureheight]
    \draw[draw=none,pattern=north east lines, pattern color=lightgray] ({0 - (\maxy +9)/14} ,\maxy) -- (0, -9) -- ({0-(\miny + 9)/14},\miny) -- (\maxx,\miny) -- (\maxx,\maxy) -- (\minx,\maxy) -- cycle; 
     \draw[draw=none,pattern=north west lines, pattern color=lightgray] (\minx,{-4 - (\minx - 0)*10}) -- (0, -4) -- (\maxx,{-4 - (\maxx - 0)*10}) -- (\maxx,\miny) -- (\maxx,\maxy) -- (\minx,\maxy) -- cycle; 
   \draw[draw=none,fill=lightgray,opacity=0.5] ({0 - (\maxy +9)/14} ,\maxy) -- (-5/4, 17/2) -- (\maxx,{-4 - (\maxx - 0)*10}) -- (\maxx,\miny) -- (\maxx,\maxy) -- cycle; 
      \foreach \y/\yname in {-4,-9}
         \draw[tick] (0,\y) node[left,ticklabel] {\footnotesize$\yname$} -- ++(1ex,0);

         \foreach \y/\yname in {{17/2}/\tfrac{17}{2}}
         \draw[tick] (0,\y) node[left,ticklabel] {\footnotesize$\yname$} -- ++(1ex,0);

       \foreach \y/\yname in {{13/6}/\tfrac{13}{6}}
          \draw[tick] (0,\y) -- ++(1ex,0) node[right,ticklabel] {\footnotesize$\yname$};
       \foreach \x/\xname in {{{-5/4}}/-\tfrac{5}{4}\phantom-,{{-9/14}}/-\tfrac{9}{14}\phantom-,{-2/5}/-\tfrac{2}{5}\phantom-}
         \draw[tick] (\x,0) node[below,ticklabel] {\footnotesize$\xname$} -- ++(0,1ex);

\draw[axis] (\minx,0) -- (\maxx,0) node[above left,axislabel] {$x^1$};
    \draw[axis] (0,\maxy) -- (0,\miny) node[above right,axislabel] {$x^2$};
     \draw (-7/12,13/6) node {$\bullet$};
     \draw (-7/12,13/6) node[right] {$z^{\mathrm{bu}}_2$};
    \draw[thick] ({0 - (\maxy +9)/14} ,\maxy) -- (-5/4, 17/2) -- (\maxx,{-4 - (\maxx - 0)*10}); 
    \draw (\maxx,\maxy) node[below left] {$\mathcal{W}^{\mathrm{bu}}_1=\mathcal{Z}^{\mathrm{buu}}_2\cap\mathcal{Z}^{\mathrm{bud}}_2$};
     \draw ({0 - (\maxy +9)/14} ,\maxy) -- (0, -9) node[sloped,below,pos=0.8] {$\mathcal{Z}^{\mathrm{buu}}_2$} -- ({0-(\miny + 9)/14},\miny);
     \draw (\minx,{-4 - (\minx - 0)*10}) -- (0, -4) -- (\maxx,{-4 - (\maxx - 0)*10}) node[sloped,above,midway] {$\mathcal{Z}^{\mathrm{bud}}_2$};
  \end{tikzpicture}
 &
      \footnotesize\begin{tikzpicture}[x=0.5/(\maxx-(\minx))*\figurewidth,y=0.5/(\maxy-(\miny))*\figureheight]
   \draw[draw=none,pattern=north east lines, pattern color=lightgray] ({0 - (\maxy +9)/14} ,\maxy) -- (-5/4, 17/2) -- (\maxx,{-4 - (\maxx - 0)*10}) -- (\maxx,\miny) -- (\maxx,\maxy) -- cycle; 
   \draw[draw=none,fill=lightgray,opacity=0.5] ({0 - (\maxy +9)/14} ,\maxy) -- (-5/4, 17/2) -- (\maxx,{-4 - (\maxx - 0)*10}) -- (\maxx,\miny) -- (\maxx,\maxy) -- cycle; 

      \foreach \y/\yname in {-4}
         \draw[tick] (0,\y) node[left,ticklabel] {\footnotesize$\yname$} -- ++(1ex,0);

         \foreach \y/\yname in {{17/2}/\tfrac{17}{2}}
         \draw[tick] (0,\y) node[left,ticklabel] {\footnotesize$\yname$} -- ++(1ex,0);

         \foreach \y/\yname in {{13/6}/\tfrac{13}{6}}
         \draw[tick] (0,\y) -- ++(1ex,0) node[right,ticklabel] {\footnotesize$\yname$};

       \foreach \x/\xname in {{-7/12}/-\tfrac{7}{12}\phantom-,{-5/4}/-\tfrac{5}{4}\phantom-,{-2/5}/-\tfrac{2}{5}\phantom-}
         \draw[tick] (\x,0) node[below,ticklabel] {\footnotesize$\xname$} -- ++(0,1ex);

      \draw[axis] (\minx,0) -- (\maxx,0) node[above left,axislabel] {$x^1$};
      \draw[axis] (0,\maxy) -- (0,\miny) node[above right,axislabel] {$x^2$};
     \draw (-7/12,13/6) node {$\bullet$};
     \draw (-7/12,13/6) node[left] {$z^{\mathrm{bu}}_2 = z^{\mathrm{b}}_1$};
      \draw[thick] ({0 - (\maxy +9)/14} ,\maxy) -- (-5/4, 17/2) -- (\maxx,{-4 - (\maxx - 0)*10}) node[sloped,above,pos=0.85] {$\mathcal{W}^{\mathrm{bu}}_1$};
     \draw (\maxx,\maxy) node[below left] {$\mathcal{V}^{\mathrm{bu}}_1=\mathcal{W}^{\mathrm{bu}}_1+\mathcal{Q}^{\mathrm{u}}_1$};
     \draw[dashed] ({0 - (\maxy-0)/14},\maxy) -- (0, 0)  node[sloped, below, pos=0.43] {lower boundary of $\mathcal{Q}^\mathrm{u}_1$} -- (\maxx, {0 - (\maxx-0)*10});
   \end{tikzpicture}
 \\
      \footnotesize\begin{tikzpicture}[x=0.5/(\maxx-(\minx))*\figurewidth,y=0.5/(\maxy-(\miny))*\figureheight]
     \draw[draw=none,pattern=north east lines, pattern color=lightgray]  ({0 - (\maxy +9)/14} ,\maxy) -- (-5/4, 17/2) -- (\maxx,{-4 - (\maxx - 0)*10}) -- (\maxx,\miny) -- (\maxx,\maxy) -- cycle; 
     \draw[draw=none,pattern=north west lines, pattern color=lightgray]  ({0 - (\maxy+6)/14},\maxy) -- (0, -6) -- (\maxx,{-6 - (\maxx - 0)*10}) -- (\maxx,\maxy) -- (\minx,\maxy) -- cycle; 
     \draw[draw=none,fill=lightgray,opacity=0.5] ({0 - (\maxy+6)/14},\maxy) -- (-1/2, 1) -- (\maxx,{-4 - (\maxx - 0)*10})-- (\maxx,\maxy) -- cycle; 
      \foreach \y/\yname in {-4,-6,1,{17/2}/\tfrac{17}{2}}
         \draw[tick] (0,\y) node[left,ticklabel] {\footnotesize$\yname$} -- ++(1ex,0);

       \foreach \y/\yname in {{13/6}/\tfrac{13}{6}}
          \draw[tick] (0,\y) -- ++(1ex,0) node[right,ticklabel] {\footnotesize$\yname$};
       \foreach \x/\xname in {{{-5/4}}/-\tfrac{5}{4}\phantom-,{-1/2}/-\tfrac{1}{2}\phantom-}
         \draw[tick] (\x,0) node[below,ticklabel] {\footnotesize$\xname$} -- ++(0,1ex);

      \draw[axis] (\minx,0) -- (\maxx,0) node[above left,axislabel] {$x^1$};
      \draw[axis] (0,\maxy) -- (0,\miny) node[above right,axislabel] {$x^2$};
      \draw ({0 - (\maxy +9)/14} ,\maxy) -- (-5/4, 17/2) -- (\maxx,{-4 - (\maxx - 0)*10}) node[sloped,above,pos=0.85] {$\mathcal{V}^{\mathrm{bu}}_1$};
     \draw ({0 - (\maxy+6)/14},\maxy) -- (0, -6) node[sloped, above, pos=0.3] {$\mathcal{X}^{\mathrm{bu}}_1$} -- (\maxx,{-6 - (\maxx - 0)*10})  ;
     \draw (-7/12,13/6) node {$\bullet$};
     \draw (-7/12,13/6) node[above] {$z^{\mathrm{b}}_1$};
      \draw[thick] ({0 - (\maxy+6)/14},\maxy) -- (-1/2, 1) -- (\maxx,{-4 - (\maxx - 0)*10});
      \draw (\maxx,\maxy) node[below left] {$\mathcal{V}^{\mathrm{bu}}_1\cap\mathcal{X}^{\mathrm{bu}}_1$};
   \end{tikzpicture}
 &
     \footnotesize\begin{tikzpicture}[x=0.5/(\maxx-(\minx))*\figurewidth,y=0.5/(\maxy-(\miny))*\figureheight]
     \draw[draw=none,pattern=north east lines, pattern color=lightgray]  ({0 - (\maxy+6)/14},\maxy) -- (-1/2, 1) -- (\maxx,{-4 - (\maxx - 0)*10})-- (\maxx,\maxy) -- cycle; 
     \draw[draw=none,pattern=north west lines, pattern color=lightgray]  ({0 - (\maxy+3)/14},\maxy) -- (0, -3) -- (\maxx, {-3 - (\maxx-0)*10}) -- (\maxx,\maxy) -- (\minx,\maxy) -- cycle; 
     \draw[draw=none,fill=lightgray,opacity=0.5]  ({0 - (\maxy+6)/14},\maxy) -- (-1/2, 1) -- (\maxx,{-4 - (\maxx - 0)*10})-- (\maxx,\maxy) -- cycle; 

        \foreach \y/\yname in {-4,1}
         \draw[tick] (0,\y) node[left,ticklabel] {\footnotesize$\yname$} -- ++(1ex,0);

         \foreach \y/\yname in {-3,{13/6}/\tfrac{13}{6}}
         \draw[tick] (0,\y) -- ++(1ex,0) node[right,ticklabel] {\footnotesize$\yname$};

       \foreach \x/\xname in {{-3/14}/-\tfrac{3}{14}\phantom-}
         \draw[tick] (\x,0) -- ++(0,1ex) node[above,ticklabel] {\footnotesize$\xname$};

       \foreach \x/\xname in {{-7/12}/-\tfrac{7}{12}\phantom-,{-2/5}/-\tfrac{2}{5}\phantom-}
         \draw[tick] (\x,0) node[below,ticklabel] {\footnotesize$\xname$} -- ++(0,1ex);

\draw[axis] (\minx,0) -- (\maxx,0) node[above left,axislabel] {$x^1$};
      \draw[axis] (0,\maxy) -- (0,\miny) node[above right,axislabel] {$x^2$};
     \draw[thick] ({0 - (\maxy+6)/14},\maxy) -- (-1/2, 1) node[sloped, below, pos=0.6] {$\mathcal{V}^{\mathrm{bu}}_1\cap\mathcal{X}^{\mathrm{bu}}_1$}-- (\maxx,{-4 - (\maxx - 0)*10})  ;
     \draw ({0 - (\maxy+3)/14},\maxy) -- (0, -3) node[sloped, above, pos=0.4] {$\mathcal{Y}^{\mathrm{bu}}_1$} -- (\maxx, {-3 - (\maxx-0)*10})  ;

     \draw (-7/12,13/6) node {$\bullet$};
     \draw (-7/12,13/6) node[above] {$z^{\mathrm{b}}_1$};
  \draw (\maxx,\maxy) node[below left] {$\mathcal{Z}^{\mathrm{bu}}_1=\conv\{\mathcal{V}^{\mathrm{bu}}_1\cap\mathcal{X}^{\mathrm{bu}}_1,\mathcal{Y}^{\mathrm{bu}}_1\}$};
   \end{tikzpicture}
   \end{tabular}
    \caption{$\mathcal{W}^{\mathrm{bu}}_1$, $\mathcal{V}^{\mathrm{bu}}_1$, $\mathcal{V}^{\mathrm{bu}}_1\cap\mathcal{X}^{\mathrm{bu}}_1$, $\mathcal{Z}^{\mathrm{bu}}_1$, $z^{\mathrm{b}}_1$, $z^{\mathrm{bu}}_2$}
    \label{fig:ex:toy:Zbu1}
   \end{figure}
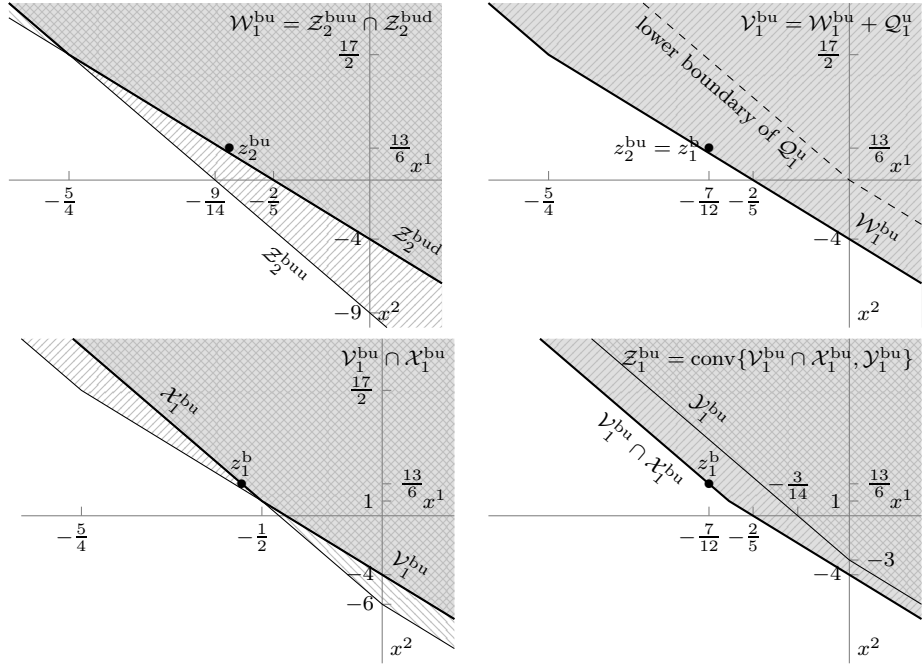

 Similarly to the seller's case, the construction of a superhedging strategy for the buyer starting from the initial endowment $\twovector{0}{-\pi_2^{\mathrm{b}}(Y,X)}$ involves two steps, namely assembling $(\psi,z^{\mathrm{b}})\in\Lambda^\mathrm{b}(Y,X)$ using the construction in the proof of Proposition~\ref{Prop:9h75hhsl}, and then converting it into a superhedging strategy $(\psi,u^{\mathrm{b}})\in\Phi^\mathrm{b}(Y,X)$ following the lines in the proof of Proposition~\ref{Prop:jg768enk}. Let us consider again the first step for the scenario $\mathrm{uu}$. Define $z^{\mathrm{b}}_0:=\twovector{0}{-\tfrac{11}{3}}$; then Figure~\ref{fig:ex:toy:Zb0} shows that $z^{\mathrm{b}}_0\in\mathcal{X}^{\mathrm{b}}_0$ but $z^{\mathrm{b}}_0\notin\mathcal{Y}^{\mathrm{b}}_0$, which leads to $\psi_0:=0$. Choosing $z^{\mathrm{b}}_1:=\twovector{-\tfrac{7}{12}}{\tfrac{13}{6}}\in\mathcal{W}^{\mathrm{b}}_0$ ensures that
 \[
    z^{\mathrm{b}}_0-\psi_0Y_0 - z^{\mathrm{b}}_1 = z^{\mathrm{b}}_0-z^{\mathrm{b}}_1\in\mathcal{Q}_0.
 \]
 Figure~\ref{fig:ex:toy:Zbu1} shows that $z^{\mathrm{b}}_1\in\mathcal{X}^{\mathrm{bu}}_1$; however $z^{\mathrm{b}}_1\notin\mathcal{Y}^{\mathrm{bu}}_1$ again leads to the choice $\psi^{\mathrm{u}}_1:=0$. Moreover $z^{\mathrm{b}}_1\in\mathcal{W}^{\mathrm{bu}}_1$, so choosing $z^{\mathrm{bu}}_2:=z^{\mathrm{b}}_1$ gives
 \[
  z^{\mathrm{b}}_1-\psi^{\mathrm{u}}_1Y^{\mathrm{u}}_1 - z^{\mathrm{bu}}_2 = 0 \in\mathcal{Q}^{\mathrm{u}}_1.
 \]
 Note finally that $z^{\mathrm{bu}}_2\in \mathcal{Z}^{\mathrm{buu}}_2=\mathcal{Y}^{\mathrm{buu}}_2$, which leads to $\psi^{\mathrm{uu}}_2:=1$. Thus this strategy corresponds to exercising the entire option at time 2 on the node $\mathrm{uu}$.
     \begin{figure}
  \newcommand{\minx}{-1.5}
  \newcommand{\maxx}{0.3}
  \newcommand{\miny}{-10}
  \newcommand{\maxy}{12}
   \centering
   \begin{tabular}{@{}rr@{}}
      \footnotesize\begin{tikzpicture}[x=0.5/(\maxx-(\minx))*\figurewidth,y=0.5/(\maxy-(\miny))*\figureheight]
     \draw[draw=none,pattern=north east lines, pattern color=lightgray] ({0 - (\maxy+6)/14},\maxy) -- (-1/2, 1) -- (\maxx,{-4 - (\maxx - 0)*10})-- (\maxx,\maxy) -- cycle; 
      \draw[draw=none,pattern=north west lines, pattern color=lightgray] (\minx,{-4/3 - (\minx - 0)*6}) -- (0, -4/3) -- (\maxx,{-4/3 - (\maxx - 0)*6}) -- (\maxx,\miny) -- (\maxx,\maxy) -- (\minx,\maxy) -- cycle; 
    \draw[draw=none,fill=lightgray,opacity=0.5] ({0 - (\maxy+6)/14},\maxy) -- (-7/12, 13/6) -- (\maxx,{-4/3 - (\maxx - 0)*6}) -- (\maxx,\maxy) -- cycle; 
       \foreach \y/\yname in {-4,{-4/3}/-\tfrac{4}{3}}
          \draw[tick] (0,\y) node[left,ticklabel] {\footnotesize$\yname$} -- ++(1ex,0);
       \foreach \y/\yname in {{13/6}/\tfrac{13}{6}}
          \draw[tick] (0,\y) -- ++(1ex,0) node[right,ticklabel] {\footnotesize$\yname$};
        \foreach \x/\xname in {{-2/9}/-\tfrac{2}{9}\phantom-}
          \draw[tick] (\x,0) -- ++(0,1ex) node[above,ticklabel] {\footnotesize$\xname$};
        \foreach \x/\xname in {{-7/12}/-\tfrac{7}{12}\phantom-}
          \draw[tick] (\x,0) node[below,ticklabel] {\footnotesize$\xname$} -- ++(0,1ex);
     \draw (-7/12,13/6) node {$\bullet$};
     \draw (-7/12,13/6) node[above] {$z^{\mathrm{b}}_1$};
 \draw[axis] (\minx,0) -- (\maxx,0) node[above left,axislabel] {$x^1$};
     \draw[axis] (0,\maxy) -- (0,\miny) node[above right,axislabel] {$x^2$};
     \draw[thick] ({0 - (\maxy+6)/14},\maxy) -- (-7/12, 13/6) -- (\maxx,{-4/3 - (\maxx - 0)*6}); 
     \draw (\maxx,\maxy) node[below left] {$\mathcal{W}^{\mathrm{b}}_0=\mathcal{Z}^{\mathrm{bu}}_1\cap\mathcal{Z}^{\mathrm{bd}}_1$};
      \draw ({0 - (\maxy+6)/14},\maxy) -- (-1/2, 1) node[sloped, above, midway] {$\mathcal{Z}^{\mathrm{bu}}_1$}-- (\maxx,{-4 - (\maxx - 0)*10});
      \draw (\minx,{-4/3 - (\minx - 0)*6}) -- (0, -4/3) -- (\maxx,{-4/3 - (\maxx - 0)*6}) node[sloped,above,midway] {$\mathcal{Z}^{\mathrm{bd}}_1$};
   \end{tikzpicture}
 &
       \footnotesize\begin{tikzpicture}[x=0.5/(\maxx-(\minx))*\figurewidth,y=0.5/(\maxy-(\miny))*\figureheight]
    \draw[draw=none,pattern=north east lines, pattern color=lightgray] ({0 - (\maxy+6)/14},\maxy) -- (-7/12, 13/6) -- (\maxx,{-4/3 - (\maxx - 0)*6}) -- (\maxx,\maxy) -- cycle; 
    \draw[draw=none,fill=lightgray,opacity=0.5]  (\minx,{13/6 - (\minx + 7/12)*10}) -- (-7/12, 13/6) -- (\maxx,{13/6 - (\maxx + 7/12)*10}) -- (\maxx,\maxy) -- (\minx,\maxy) -- cycle; 
       \foreach \y/\yname in {{-4/3}/-\tfrac{4}{3},{-11/3}/-\tfrac{11}{3}}
          \draw[tick] (0,\y) node[left,ticklabel] {\footnotesize$\yname$} -- ++(1ex,0);
       \foreach \y/\yname in {{13/6}/\tfrac{13}{6}}
          \draw[tick] (0,\y) -- ++(1ex,0) node[right,ticklabel] {\footnotesize$\yname$};

        \foreach \x/\xname in {{-7/12}/-\tfrac{7}{12}\phantom-,{-11/30}/-\tfrac{11}{30}\phantom-}
          \draw[tick] (\x,0) node[below,ticklabel] {\footnotesize$\xname$} -- ++(0,1ex);
 \draw[axis] (\minx,0) -- (\maxx,0) node[above left,axislabel] {$x^1$};
     \draw[axis] (0,\maxy) -- (0,\miny) node[above right,axislabel] {$x^2$};
       \draw ({0 - (\maxy+6)/14},\maxy) -- (-7/12, 13/6) -- (\maxx,{-4/3 - (\maxx - 0)*6}) node[sloped,above,pos=0.8] {$\mathcal{W}^{\mathrm{b}}_0$};
       \draw[thick] (\minx,{13/6 - (\minx + 7/12)*10}) -- (-7/12, 13/6) -- (\maxx,{13/6 - (\maxx + 7/12)*10});
      \draw (\maxx,\maxy) node[below left] {$\mathcal{V}^{\mathrm{b}}_0=\mathcal{W}^{\mathrm{b}}_0+\mathcal{Q}_0$};
     \draw (-7/12,13/6) node {$\bullet$};
     \draw (-7/12,13/6) node[above] {$z^{\mathrm{b}}_1$};
     \draw (0,-11/3) node {$\bullet$};
     \draw (0,-11/3) node[right] {$z^{\mathrm{b}}_0$};
    \draw[dashed]  ({0 - (\maxy-0)/10},\maxy)-- (0, 0)   node[sloped, above, midway] {lower boundary of $\mathcal{Q}_0$}-- (\maxx, {0 - (\maxx-0)*10});
    \end{tikzpicture}
 \\
       \footnotesize\begin{tikzpicture}[x=0.5/(\maxx-(\minx))*\figurewidth,y=0.5/(\maxy-(\miny))*\figureheight]
      \draw[draw=none,pattern=north east lines, pattern color=lightgray]  (\minx,{13/6 - (\minx + 7/12)*10}) -- (-7/12, 13/6) -- (\maxx,{13/6 - (\maxx + 7/12)*10}) -- (\maxx,\maxy) -- (\minx,\maxy) -- cycle; 
      \draw[draw=none,pattern=north west lines, pattern color=lightgray]  (\minx, {-5 - (\minx-0)*10}) -- (0, -5) -- (\maxx,{-5 - (\maxx - 0)*10}) -- (\maxx,\maxy) -- (\minx,\maxy) -- cycle; 
      \draw[draw=none,fill=lightgray,opacity=0.5] (\minx,{13/6 - (\minx + 7/12)*10}) -- (-7/12, 13/6) -- (\maxx,{13/6 - (\maxx + 7/12)*10}) -- (\maxx,\maxy) -- (\minx,\maxy) -- cycle; 
       \foreach \y/\yname in {-5}
          \draw[tick] (0,\y) node[left,ticklabel] {\footnotesize$\yname$} -- ++(1ex,0);
       \foreach \y/\yname in {{-11/3}/-\tfrac{11}{3}}
          \draw[tick] (0,\y) -- ++(1ex,0) node[right,ticklabel] {\footnotesize$\yname$};

        \foreach \x/\xname in {{-1/2}/-\tfrac{1}{2}\phantom-}
          \draw[tick] (\x,0) node[below,ticklabel] {\footnotesize$\xname$} -- ++(0,1ex);
        \foreach \x/\xname in {{-11/30}/-\tfrac{11}{30}\phantom-}
          \draw[tick] (\x,0) -- ++(0,1ex) node[above,ticklabel] {\footnotesize$\xname$};
 \draw[axis] (\minx,0) -- (\maxx,0) node[above left,axislabel] {$x^1$};
     \draw[axis] (0,\maxy) -- (0,\miny) node[above right,axislabel] {$x^2$};
     \draw (0,-11/3) node {$\bullet$};
     \draw (0,-11/3) node[above] {$z^{\mathrm{b}}_0$};
       \draw[thick] (\minx,{13/6 - (\minx + 7/12)*10}) -- (-7/12, 13/6) node[sloped,above,pos=0.7] {$\mathcal{V}^{\mathrm{b}}_0$} -- (\maxx,{13/6 - (\maxx + 7/12)*10});
     \draw (\minx, {-5 - (\minx-0)*10}) -- (0, -5)  node[sloped, below, pos=0.4] {$\mathcal{X}^{\mathrm{b}}_0$} -- (\maxx,{-5 - (\maxx - 0)*10});
       \draw (\maxx,\maxy) node[below left] {$\mathcal{V}^{\mathrm{b}}_0\cap\mathcal{X}^{\mathrm{b}}_0$};
    \end{tikzpicture}
 &
      \footnotesize\begin{tikzpicture}[x=0.5/(\maxx-(\minx))*\figurewidth,y=0.5/(\maxy-(\miny))*\figureheight]
      \draw[draw=none,pattern=north east lines, pattern color=lightgray] (\minx,{13/6 - (\minx + 7/12)*10}) -- (-7/12, 13/6) -- (\maxx,{13/6 - (\maxx + 7/12)*10}) -- (\maxx,\maxy) -- (\minx,\maxy) -- cycle; 
      \draw[draw=none,pattern=north west lines, pattern color=lightgray] ({0 - (\maxy-0)/10},\maxy) -- (0, 0) -- (\maxx, {0 - (\maxx-0)*10}) -- (\maxx,\maxy) -- cycle; 
      \draw[draw=none,fill=lightgray,opacity=0.5] (\minx,{13/6 - (\minx + 7/12)*10}) -- (-7/12, 13/6) -- (\maxx,{13/6 - (\maxx + 7/12)*10}) -- (\maxx,\maxy) -- (\minx,\maxy) -- cycle; 
       \foreach \y/\yname in {{-11/3}/-\tfrac{11}{3}}
          \draw[tick] (0,\y) -- ++(1ex,0) node[right,ticklabel] {\footnotesize$\yname$};

        \foreach \x/\xname in {{-11/30}/-\tfrac{11}{30}\phantom-}
          \draw[tick] (\x,0) -- ++(0,1ex) node[above,ticklabel] {\footnotesize$\xname$};
 \draw[axis] (\minx,0) -- (\maxx,0) node[above left,axislabel] {$x^1$};
     \draw[axis] (0,\maxy) -- (0,\miny) node[above right,axislabel] {$x^2$};

       \draw[thick] (\minx,{13/6 - (\minx + 7/12)*10}) -- (-7/12, 13/6) node[sloped,above,pos=0.7] {$\mathcal{V}^{\mathrm{b}}_0\cap\mathcal{X}^{\mathrm{b}}_0$} -- (\maxx,{13/6 - (\maxx + 7/12)*10});
    \draw  ({0 - (\maxy-0)/10},\maxy) -- (0, 0)  node[sloped, above, pos=0.6] {$\mathcal{Y}^{\mathrm{b}}_0$} -- (\maxx, {0 - (\maxx-0)*10});
     \draw (0,-11/3) node {$\bullet$};
     \draw (0,-11/3) node[below] {$z^{\mathrm{b}}_0$};
   \draw (\maxx,\maxy) node[below left] {$\mathcal{Z}^{\mathrm{b}}_0=\conv\{\mathcal{V}^{\mathrm{b}}_0\cap\mathcal{X}^{\mathrm{b}}_0,\mathcal{Y}^{\mathrm{b}}_0\}$};
    \end{tikzpicture}
   \end{tabular}
    \caption{$\mathcal{W}^{\mathrm{b}}_0$, $\mathcal{V}^{\mathrm{b}}_0$, $\mathcal{V}^{\mathrm{b}}_0\cap\mathcal{X}^{\mathrm{b}}_0$, $\mathcal{Z}^{\mathrm{b}}_0$, $z^{\mathrm{b}}_0$, $z^{\mathrm{b}}_1$}
    \label{fig:ex:toy:Zb0}
   \end{figure}
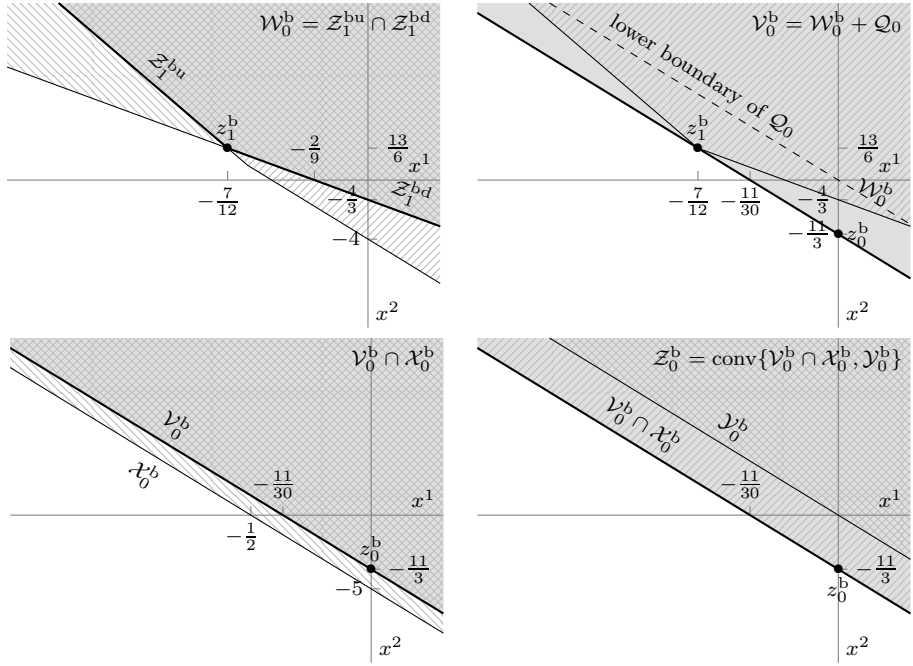

\section{Appendix: proofs\label{Sect:Appendix}}\label{sec:7}

\begin{proof}
[Proof of Theorem~\ref{Thm:fjusab76}]By Definition~\ref{Def:nnf65uwsu48},%
\[
\pi_{j}^{\mathrm{a}}(Y,X)=\inf\left\{  x\in\mathbb{R}\,|\,\exists(\phi
,u)\in\Phi^{\mathrm{a}}(Y,X):xe^{j}=u_{0}\right\}  .
\]
Hence, according to Proposition~\ref{Prop:7d4hnk0a},%
\[
\pi_{j}^{\mathrm{a}}(Y,X)=\inf\left\{  x\in\mathbb{R}\,|\,\exists(\phi
,u)\in\Lambda^{\mathrm{a}}(Y,X):xe^{j}=u_{0}\right\}
\]
and, by Proposition~\ref{Prop:7djhb78a},%
\[
\pi_{j}^{\mathrm{a}}(Y,X)=\inf\left\{  x\in\mathbb{R}\,|\,xe^{j}\in
\mathcal{Z}_{0}^{\mathrm{a}}\right\}  .
\]
Because $\mathcal{Z}_{0}^{\mathrm{a}}$ is a polyhedral set, it is closed. It
follows that $\left\{  x\in\mathbb{R}\,|\,xe^{j}\in\mathcal{Z}^{\mathrm{a}}_{0}\right\}  $
is closed. It is non-empty and bounded below because $xe^{j}\in\mathcal{Z}%
_{0}^{\mathrm{a}}$ for any $x\in\mathbb{R}$ large enough, and $xe^{j}%
\notin\mathcal{Z}_{0}^{\mathrm{a}}$ for any $x\in\mathbb{R}$ small enough. It
follows that the infimum is in fact a minimum.\bigskip
\end{proof}

\begin{proof}
[Proof of Proposition~\ref{Prop:7djhb78a}]Let $a\in\mathcal{Z}_{0}%
^{\mathrm{a}}$. We construct a mixed stopping time $\phi\in\mathcal{X}$ together with the corresponding process~$\phi^{\ast}$ and a
strategy $z\in\Phi$ by induction. First we put $\phi_{0}^{\ast}:=1$ and
$z_{0}:=a$. Clearly, $z_{0}\in\phi_{0}^{\ast}\mathcal{Z}_{0}^{\mathrm{a}}$.
Now suppose that for some $t=0,\ldots,T-1$ we have constructed~$z_{t}$
and~$\phi_{t}^{\ast}$ such that $z_{t}\in\phi_{t}^{\ast}\mathcal{Z}%
_{t}^{\mathrm{a}}$. It follows that $z_{t}\in\phi_{t}^{\ast}\mathcal{Y}%
_{t}^{\mathrm{a}}$, hence%
\[
z_{t}-\phi_{t}^{\ast}Y_{t}\in\mathcal{Q}_{t}.
\]
It also follows that $z_{t}\in\phi_{t}^{\ast}\operatorname*{conv}%
\{\mathcal{V}_{t}^{\mathrm{a}},\mathcal{X}_{t}^{\mathrm{a}}\}$, hence there
exist $\lambda_{t}\in\lbrack0,1]$, $v_{t}\in\mathcal{V}_{t}^{\mathrm{a}}$ and
$x_{t}\in\mathcal{X}_{t}^{\mathrm{a}}$ such that $z_{t}=\phi_{t}^{\ast}\left(
(1-\lambda_{t})v_{t}+\lambda_{t}x_{t}\right)  $. We put $\phi_{t}:=\phi
_{t}^{\ast}\lambda_{t}$, and then $\phi_{t+1}^{\ast}:=\phi_{t}^{\ast}-\phi
_{t}=\phi_{t}^{\ast}(1-\lambda_{t})$, so $z_{t}=\phi_{t+1}^{\ast}v_{t}%
+\phi_{t}x_{t}$. Since $x_{t}\in\mathcal{X}_{t}^{\mathrm{a}}$ and $v_{t}%
\in\mathcal{V}_{t}^{\mathrm{a}}$, it follows that $x_{t}-X_{t}\in
\mathcal{Q}_{t}$ and there is $z_{t+1}\in\phi_{t+1}^{\ast}\mathcal{W}%
_{t}^{\mathrm{a}}$ such that $\phi_{t+1}^{\ast}v_{t}-z_{t+1}\in\mathcal{Q}%
_{t}$. As a result,%
\begin{align*}
z_{t}-\phi_{t}X_{t}-z_{t+1}  &  =\phi_{t+1}^{\ast}v_{t}+\phi_{t}x_{t}-\phi
_{t}X_{t}-z_{t+1}\\
&  =\phi_{t+1}^{\ast}v_{t}-z_{t+1}+\phi_{t}\left(  x_{t}-X_{t}\right)
\in\mathcal{Q}_{t}.
\end{align*}
Since $\mathcal{W}_{t}^{\mathrm{a}}\subseteq\mathcal{Z}_{t+1}^{\mathrm{a}}$,
it also follows that $z_{t+1}\in\phi_{t+1}^{\ast}\mathcal{Z}_{t+1}%
^{\mathrm{a}}$. Finally, given that $z_{T}\in\phi_{T}^{\ast}\mathcal{Z}%
_{T}^{\mathrm{a}}$, we get $z_{T}\in\phi_{T}^{\ast}\mathcal{Y}_{T}%
^{\mathrm{a}}$, so $z_{T}-\phi_{T}^{\ast}Y_{T}\in\mathcal{Q}_{T}$, and we put
$\phi_{T}:=\phi_{T}^{\ast}$, $\phi_{T+1}^{\ast}:=0$ and $z_{T+1}:=0$. We have
constructed $(\phi,z)\in\Lambda^{\mathrm{a}}(Y,X)$ such that $a=z_{0}$.

Conversely, we take any $(\phi,z)\in\Lambda^{\mathrm{a}}(Y,X)$, and want to
show that $z_{0}\in\mathcal{Z}_{0}^{\mathrm{a}}$. More generally, we will show
by backward induction that for each $t=0,\ldots,T$%
\begin{equation}
z_{t}\in\left\{
\begin{array}
[c]{ll}%
\phi_{t}^{\ast}\mathcal{Z}_{t}^{\mathrm{a}} & \text{on }\{\phi_{t}^{\ast}>0\}\\
\mathcal{Q}_{t} & \text{on }\{\phi_{t}^{\ast}=0\}
\end{array}
\right.  . \label{eq:f754iewndfg}%
\end{equation}
Since $z_{T}-\phi_{T}^{\ast}Y_{T}=z_{T}-\phi_{T}Y_{T}\in\mathcal{Q}_{T}$, it
follows that%
\[
z_{T}\in\phi_{T}^{\ast}Y_{T}+\mathcal{Q}_{T}=\left\{
\begin{array}
[c]{ll}%
\phi_{T}^{\ast}\mathcal{Y}_{T}^{\mathrm{a}} & \text{on }\{\phi_{T}^{\ast}>0\}\\
\mathcal{Q}_{T} & \text{on }\{\phi_{T}^{\ast}=0\}
\end{array}
\right.  =\left\{
\begin{array}
[c]{ll}%
\phi_{T}^{\ast}\mathcal{Z}_{T}^{\mathrm{a}} & \text{on }\{\phi_{T}^{\ast}>0\}\\
\mathcal{Q}_{T} & \text{on }\{\phi_{T}^{\ast}=0\}
\end{array}
\right.  .
\]
Next, suppose that (\ref{eq:f754iewndfg}) holds for some $t=1,\ldots,T$. Since
$z$ is predictable, it follows that $z_{t}\in\mathcal{L}_{t-1}$, so%
\[
z_{t}\in\left\{
\begin{array}
[c]{ll}%
\phi_{t}^{\ast}(\mathcal{Z}_{t}^{\mathrm{a}}\cap\mathcal{L}_{t-1}) & 
\text{on }\{\phi_{t}^{\ast}>0\}\\
\mathcal{Q}_{t}\cap\mathcal{L}_{t-1} & \text{on }\{\phi_{t}^{\ast}=0\}
\end{array}
\right.  \subseteq\left\{
\begin{array}
[c]{ll}%
\phi_{t}^{\ast}\mathcal{W}_{t-1}^{\mathrm{a}} & \text{on }\{\phi_{t}^{\ast}>0\}\\
\mathcal{Q}_{t-1} & \text{on }\{\phi_{t}^{\ast}=0\}
\end{array}
\right.  .
\]
Hence, using%
\[
z_{t-1}-\phi_{t-1}X_{t-1}-z_{t}\in\mathcal{Q}_{t-1},
\]
we obtain%
\setlongestterm{\phi_{t}^{\ast}\mathcal{W}_{t-1}^{\mathrm{a}}+\mathcal{Q}_{t-1}}
\begin{align*}
z_{t-1}-\phi_{t-1}X_{t-1}  &  \in z_{t}+\mathcal{Q}_{t-1}\\
&  \subseteq\left\{
\begin{array}
[c]{ll}%
\phi_{t}^{\ast}\mathcal{W}_{t-1}^{\mathrm{a}}+\mathcal{Q}_{t-1} & 
\text{on }\{\phi_{t}^{\ast}>0\}\\
\mathcal{Q}_{t-1}+\mathcal{Q}_{t-1} & \text{on }\{\phi_{t}^{\ast}=0\}
\end{array}
\right. \\
&  =\left\{
\begin{array}
[c]{ll}%
\adjusttolongestterm{\phi_{t}^{\ast}\mathcal{V}_{t-1}^{\mathrm{a}}} & 
\text{on }\{\phi_{t}^{\ast}>0\}\\
\mathcal{Q}_{t-1} & \text{on }\{\phi_{t}^{\ast}=0\}
\end{array}
\right.  .
\end{align*}
It follows that%
\setlongestterm{\phi_{t}^{\ast}\operatorname*{conv}\left\{  \mathcal{V}_{t-1}^{\mathrm{a}%
},\mathcal{X}_{t-1}^{\mathrm{a}}\right\}}
\begin{align*}
z_{t-1}  &  \in\left\{
\begin{array}
[c]{ll}%
\adjusttolongestterm[+\innotis]{\phi_{t-1}X_{t-1}+\phi_{t}^{\ast}\mathcal{V}_{t-1}^{\mathrm{a}}} & \text{on }\{\phi_{t}^{\ast}>0\}\\
\phi_{t-1}X_{t-1}+\mathcal{Q}_{t-1} & \text{on }\{\phi_{t}^{\ast}=0\}
\end{array}
\right. \\
&  =\left\{
\begin{array}
[c]{ll}%
\adjusttolongestterm[-\innotis]{\phi_{t}^{\ast}\mathcal{V}_{t-1}^{\mathrm{a}}+\phi_{t-1}\mathcal{X}%
_{t-1}^{\mathrm{a}}} & \text{on }\{\phi_{t}^{\ast}>0\}\\
\phi_{t-1}\mathcal{X}_{t-1}^{\mathrm{a}} & \text{on }\{\phi_{t}^{\ast}=0,
\phi_{t-1}>0\}\\
\mathcal{Q}_{t-1} & \text{on }\{\phi_{t}^{\ast}=0, \phi_{t-1}=0\}
\end{array}
\right. \\
&  =\left\{
\begin{array}
[c]{ll}%
\adjusttolongestterm{\phi_{t}^{\ast}\mathcal{V}_{t-1}^{\mathrm{a}}+\phi_{t-1}\mathcal{X}%
_{t-1}^{\mathrm{a}}} & \text{on }\{\phi_{t-1}^{\ast}>0\}\\
\mathcal{Q}_{t-1} & \text{on }\{\phi_{t-1}^{\ast}=0\}
\end{array}
\right. \\
&  \subseteq\left\{
\begin{array}
[c]{ll}%
\phi_{t}^{\ast}\operatorname*{conv}\left\{  \mathcal{V}_{t-1}^{\mathrm{a}%
},\mathcal{X}_{t-1}^{\mathrm{a}}\right\}  & \text{on }\{\phi_{t-1}^{\ast}>0\}\\
\mathcal{Q}_{t-1} & \text{on }\{\phi_{t-1}^{\ast}=0\}
\end{array}
\right.  .
\end{align*}
Moreover, sice $z_{t-1}-\phi_{t-1}^{\ast}Y_{t-1}\in\mathcal{Q}_{t-1}$, it
follows that%
\[
z_{t-1}\in\phi_{t-1}^{\ast}Y_{t-1}+\mathcal{Q}_{t-1}=\left\{
\begin{array}
[c]{ll}%
\phi_{t}^{\ast}\mathcal{Y}_{t-1}^{\mathrm{a}} & \text{on }\{\phi_{t-1}^{\ast
}>0\}\\
\mathcal{Q}_{t-1} & \text{on }\{\phi_{t-1}^{\ast}=0\}
\end{array}
\right.  .
\]
As a result,%
\setlongestterm{\phi_{t}^{\ast}\operatorname*{conv}\left\{  \mathcal{V}_{t-1}^{\mathrm{a}%
},\mathcal{X}_{t-1}^{\mathrm{a}}\right\}  \cap\phi_{t}^{\ast}\mathcal{Y}%
_{t-1}^{\mathrm{a}}}
\begin{align*}
z_{t-1}  &  \in\left\{
\begin{array}
[c]{ll}%
\phi_{t}^{\ast}\operatorname*{conv}\left\{  \mathcal{V}_{t-1}^{\mathrm{a}%
},\mathcal{X}_{t-1}^{\mathrm{a}}\right\}  \cap\phi_{t}^{\ast}\mathcal{Y}%
_{t-1}^{\mathrm{a}} & \text{on }\{\phi_{t-1}^{\ast}>0\}\\
\mathcal{Q}_{t-1} & \text{on }\{\phi_{t-1}^{\ast}=0\}
\end{array}
\right. \\
&  =\left\{
\begin{array}
[c]{ll}%
\adjusttolongestterm[-\innotis]{\phi_{t}^{\ast}\mathcal{Z}_{t-1}^{\mathrm{a}}} & \text{on }\{\phi_{t-1}^{\ast}>0\}\\
\mathcal{Q}_{t-1} & \text{on }\{\phi_{t-1}^{\ast}=0\}
\end{array}
\right.  ,
\end{align*}
which completes the proof.\bigskip
\end{proof}

\begin{proof}
[Proof of Proposition~\ref{Prop:7d4hnk0a}]Suppose that $(\phi,z)\in
\Lambda^{\mathrm{a}}(Y,X)$. Then, for each $t=0,\ldots,T-1$%
\[
z_{t}-\phi_{t}X_{t}-z_{t+1}\in\mathcal{Q}_{t},
\]
so there is a liquidation strategy $y_{t+1}^{t},\ldots,y_{T+1}^{t}$ starting
from $z_{t}-\phi_{t}Y_{t}-z_{t+1}$ at time~$t$. We also put $y_{T+1}^{T}:=0$
for notational convenience. Moreover,
\[
z_{t}-\phi_{t}^{\ast}Y_{t}\in\mathcal{Q}_{t}\quad\text{for each }%
t=0,\ldots,T,
\]
so there is a liquidation strategy $x_{t+1}^{t},\ldots,x_{T+1}^{t}$ starting
from $z_{t}-\phi_{t}^{\ast}Y_{t}$ at time~$t$. For each $\psi\in\mathcal{X}$
we put%
\begin{align*}
u_{0}^{\psi}  &  :=z_{0},\\
u_{t}^{\psi}  &  :=\psi_{t}^{\ast}z_{t}+\sum_{s=0}^{t-1}\psi_{s+1}^{\ast}%
y_{t}^{s}+\sum_{s=0}^{t-1}\psi_{s}x_{t}^{s}\quad\text{for }t=1,\ldots,T+1.
\end{align*}
This defines $u:\mathcal{X}\rightarrow\Phi$, which satisfies the
non-anticipation condition~(\ref{eq:ug75hsdj6}). Moreover, for each $\psi
\in\mathcal{X}$ and for each $t=0,\ldots,T$,%
\begin{align*}
&  u_{t}^{\psi}-G_{t}^{\phi,\psi}-u_{t+1}^{\psi}\\
&  =\psi_{t}^{\ast}z_{t}+\sum_{s=0}^{t-1}\psi_{s+1}^{\ast}y_{t}^{s}+\sum
_{s=0}^{t-1}\psi_{s}x_{t}^{s}-\psi_{t}\phi_{t}^{\ast}Y_{t}-\psi_{t+1}^{\ast
}\phi_{t}X_{t}\\
&  \quad-\psi_{t+1}^{\ast}z_{t+1}-\sum_{s=0}^{t}\psi_{s+1}^{\ast}y_{t+1}%
^{s}-\sum_{s=0}^{t}\psi_{s}x_{t+1}^{s}\\
&  =\psi_{t+1}^{\ast}\left(  z_{t}-\phi_{t}X_{t}-z_{t+1}-y_{t+1}^{t}\right)
+\psi_{t}\left(  z_{t}-\phi_{t}^{\ast}Y_{t}-x_{t+1}^{t}\right) \\
&  \quad+\sum_{s=0}^{t-1}\psi_{s+1}^{\ast}\left(  y_{t}^{s}-y_{t+1}%
^{s}\right)  +\sum_{s=0}^{t-1}\psi_{s}\left(  x_{t}^{s}-x_{t+1}^{s}\right) \\
&  \in\psi_{t+1}^{\ast}\mathcal{K}_{t}+\psi_{t}\mathcal{K}_{t}+\sum
_{s=0}^{t-1}\psi_{s+1}^{\ast}\mathcal{K}_{t}+\sum_{s=0}^{t-1}\psi
_{s}\mathcal{K}_{t}\subseteq\mathcal{K}_{t}.
\end{align*}
This means that $(\phi,u)\in\Phi^{\mathrm{a}}(Y,X)$, with $z_{0}=u_{0}$.

Conversely, suppose that $(\phi,u)\in\Phi^{\mathrm{a}}(Y,X)$. Then we put%
\[
z:=u^{\chi^{T}}.
\]
It follows that for each $t=0,\ldots,T-1$%
\[
z_{t}-\phi_{t}X_{t}-z_{t+1}=u_{t}^{\chi^{T}}+G_{t}^{\phi,\chi^{T}}%
-u_{t+1}^{\chi^{T}}\in\mathcal{K}_{t}\subseteq\mathcal{Q}_{t}%
\]
since%
\[
G_{t}^{\phi,\chi^{T}}=\chi_{t}^{T}\phi_{t}^{\ast}Y_{t}+\chi_{t+1}^{T\ast}\phi
_{t}X_{t}=\phi_{t}X_{t}.
\]
Next, take any $t=0,\ldots,T$. Then $\chi_{s}^{T}=\chi_{s}^{t}=0$ for each
$s=0,\ldots,t-1$, and because $u$ satisfies the non-anticipation
condition~(\ref{eq:ug75hsdj6}), we have $z_{t}=u_{t}^{\chi^{T}}=u_{t}%
^{\chi^{t}}$. Since $\chi_{t}^{t}=1$, $\chi_{t+1}^{t\ast}=0$ and%
\[
G_{t}^{\phi,\chi^{t}}=\chi_{t}^{t}\phi_{t}^{\ast}Y_{t}+\chi_{t+1}^{t\ast}%
\phi_{t}X_{t}=\phi_{t}^{\ast}Y_{t},
\]
it means that%
\begin{equation}
z_{t}-\phi_{t}^{\ast}Y_{t}-u_{t+1}^{\chi^{t}}=u_{t}^{\chi^{t}}+G_{t}%
^{\phi,\chi^{t}}-u_{t+1}^{\chi^{t}}\in\mathcal{K}_{t}\subseteq\mathcal{Q}_{t}.
\label{eq:fd654jhuswbhf}%
\end{equation}
Moreover, for each $s=t+1,\ldots,T$ we have $\chi_{s}^{t}=\chi_{s+1}^{t\ast
}=0$ and%
\[
G_{s}^{\phi,\chi^{t}}=\chi_{s}^{t}\phi_{s}^{\ast}Y_{s}+\chi_{s+1}^{t\ast}%
\phi_{t}X_{t}=0,
\]
hence%
\[
u_{s}^{\chi^{t}}-u_{s+1}^{\chi^{t}}=u_{s}^{\chi^{t}}+G_{s}^{\phi,\chi^{t}%
}-u_{s+1}^{\chi^{t}}\in\mathcal{K}_{s}\subseteq\mathcal{Q}_{s}.
\]
We can verify by backward induction that $u_{s+1}^{\chi^{t}}\in\mathcal{Q}%
_{s}$ for each $s=t,\ldots,T$. Clearly, $u_{T+1}^{\chi^{t}}=0\in
\mathcal{Q}_{T}$. Now suppose that $u_{s+1}^{\chi^{t}}\in\mathcal{Q}_{s}$ for
some $s=t+1,\ldots,T$. It follows that $u_{s}^{\chi^{t}}=(u_{s}^{\chi^{t}%
}-u_{s+1}^{\chi^{t}})+u_{s+1}^{\chi^{t}}\in\mathcal{Q}_{s}+\mathcal{Q}%
_{s}=\mathcal{Q}_{s}$. By predictability, $u_{s}^{\chi^{t}}\in\mathcal{L}%
_{s-1}$, so we can conclude that $u_{s}^{\chi^{t}}\in\mathcal{Q}_{s}%
\cap\mathcal{L}_{s-1}\subseteq\mathcal{Q}_{s-1}$, which completes the backward
induction argument. In particular, we have shown that $u_{t+1}^{\chi^{t}}%
\in\mathcal{Q}_{t}$. Together with~(\ref{eq:fd654jhuswbhf}), this shows that%
\[
z_{t}-\phi_{t}^{\ast}Y_{t}\in\mathcal{Q}_{t}%
\]
for each $t=0,\ldots,T$. As a result, $(\phi,z)\in\Lambda^{\mathrm{a}}(Y,X)$
with $z_{0}=u_{0}$, which completes the proof.\bigskip
\end{proof}

\begin{proof}
[Proof of Lemma~\ref{Lem:8f5sgdjf}]Take any $u\in\Phi$ such that $(\phi
,u)\in\Phi^{\mathrm{a}}(Y,X)$. Observe that%
\begin{align*}
Q_{\phi,t}  &  =Q_{\phi,\chi^{t}}=\sum_{s=0}^{T}\sum_{u=0}^{T}\phi_{s}\chi
_{u}^{t}Q_{s,u}=\sum_{s=0}^{T}\phi_{s}\mathbf{1}_{\left\{  s\geq t\right\}
}Y_{t}+\sum_{s=0}^{T}\phi_{s}\mathbf{1}_{\left\{  s<t\right\}  }X_{s}\\
&  =\phi_{t}^{\ast}Y_{t}+\sum_{s=0}^{t-1}\phi_{s}X_{s}%
\end{align*}
and define $z:\mathcal{X}\rightarrow\Phi$ such that%
\begin{equation}
z_{t}^{\psi}:=u_{t}^{\psi}+\psi_{t}^{\ast}\sum_{s=0}^{t-1}\phi_{s}X_{s}
\label{eq:gf8dhsge}%
\end{equation}
for any $\psi\in\mathcal{X}$ and any $t=0,\ldots,T+1$.\ Then~$z$ satisfies the
non-anticipation condition~(\ref{eq:dg453hd7ah}), $z_{0}=u_{0}$, and it also
satisfies the rebalancing condition~(\ref{eq:jgt775g10}) for an American
option with payoff process~$Q_{\phi,\,\cdot\,}$ since for any $\psi
\in\mathcal{X}$ and any $t=0,\ldots,T$%
\begin{align*}
&  z_{t}^{\psi}-\psi_{t}Q_{\phi,t}-z_{t+1}^{\psi}\\
&  =\left(  u_{t}^{\psi}+\psi_{t}^{\ast}\sum_{s=0}^{t-1}\phi_{s}X_{s}\right)
-\psi_{t}\left(  \phi_{t}^{\ast}Y_{t}+\sum_{s=0}^{t-1}\phi_{s}X_{s}\right)
-\left(  u_{t+1}^{\psi}+\psi_{t+1}^{\ast}\sum_{s=0}^{t}\phi_{s}X_{s}\right) \\
&  =u_{t}^{\psi}-\psi_{t}\phi_{t}^{\ast}Y_{t}-\psi_{t+1}^{\ast}\phi_{t}%
X_{t}-u_{t+1}^{\psi}\\
&  =u_{t}^{\psi}-G_{t}^{\phi,\psi}-u_{t+1}^{\psi}\in\mathcal{K}_{t}.
\end{align*}
Conversely, take any $z\in\Psi^{\mathrm{a}}(Q_{\phi,\,\cdot\,})$ and define
$u:\mathcal{X}\rightarrow\Phi$ such that%
\[
u_{t}^{\psi}:=z_{t}^{\psi}-\psi_{t}^{\ast}\sum_{s=0}^{t-1}\phi_{s}X_{s}%
\]
for any $\psi\in\mathcal{X}$ and any $t=0,\ldots,T+1$. Then $u$ satisfies the
non-anticipation condition~(\ref{eq:ug75hsdj6}), $u_{0}=z_{0}$, and%
\begin{align*}
&  u_{t}^{\psi}-G_{t}^{\phi,\psi}-u_{t+1}^{\psi}\\
&  =\left(  z_{t}^{\psi}-\psi_{t}^{\ast}\sum_{s=0}^{t-1}\phi_{s}X_{s}\right)
-\left(  \psi_{t}\phi_{t}^{\ast}Y_{t}+\psi_{t+1}^{\ast}\phi_{t}X_{t}\right)
-\left(  z_{t+1}^{\psi}-\psi_{t+1}^{\ast}\sum_{s=0}^{t}\phi_{s}X_{s}\right) \\
&  =z_{t}^{\psi}-\psi_{t}\left(  \phi_{t}^{\ast}Y_{t}+\sum_{s=0}^{t-1}\phi
_{s}X_{s}\right)  -z_{t+1}^{\psi}\\
&  =z_{t}^{\psi}-\psi_{t}Q_{\phi,t}-z_{t+1}^{\psi}\in\mathcal{K}_{t}%
\end{align*}
for any $\psi\in\mathcal{X}$ and any $t=0,\ldots,T$, so the rebalancing
condition~(\ref{eq:f7enasg0}) holds. The lemma follows because
(\ref{eq:gf8dhsge}) defines a one-to-one map between strategies $z\in
\Psi^{\mathrm{a}}(Q_{\phi,\,\cdot\,})$ and strategies $u$ such that
$(\phi,u)\in\Phi^{\mathrm{a}}(Y,X)$ with $u_{0}=z_{0}$.\bigskip
\end{proof}

\begin{proof}
[Proof of Theorem~\ref{Thm:eudn9y65}]According to
Definition~\ref{Def:nnf65uwsu48},%
\begin{equation}
\pi_{j}^{\mathrm{a}}(Y,X)=\inf\left\{  x\in\mathbb{R}\,|\,\exists(\phi
,u)\in\Phi^{\mathrm{a}}(Y,X):xe^{j}=u_{0}\right\}  . \label{eq:nnf755efdvd64}%
\end{equation}
By Theorem~\ref{Thm:fjusab76}, $\pi_{j}^{\mathrm{a}}(Y,X)e^{j}\in
\mathcal{Z}_{0}^{\mathrm{a}}$. Hence, by Proposition~\ref{Prop:7djhb78a},
there is a $(\phi,z)\in\Lambda^{\mathrm{a}}(Y,X)$ such that $\pi
_{j}^{\mathrm{a}}(Y,X)e^{j}=z_{0}$. It follows by
Proposition~\ref{Prop:7d4hnk0a} that there is a $(\phi,u)\in\Phi^{\mathrm{a}%
}(Y,X)$ such that $\pi_{j}^{\mathrm{a}}(Y,X)e^{j}=u_{0}$, so the infimum
in~(\ref{eq:nnf755efdvd64}) is in fact a minimum,%
\[
\pi_{j}^{\mathrm{a}}(Y,X)=\min\left\{  x\in\mathbb{R}\,|\,\exists(\phi
,u)\in\Phi^{\mathrm{a}}(Y,X):xe^{j}=u_{0}\right\}  .
\]
As a result, according to Lemma~\ref{Lem:8f5sgdjf} and
Definition~\ref{Def:hfy468ahd564vk},%
\begin{align*}
\pi_{j}^{\mathrm{a}}(Y,X)  &  =\min\left\{  x\in\mathbb{R}\,|\,\exists\phi
\in\chi\exists z\in\Psi^{\mathrm{a}}(Q_{\phi,\,\cdot\,}):xe^{j}=z_{0}\right\}
\\
&  =\min_{\phi\in\chi}\inf\left\{  x\in\mathbb{R}\,|\exists z\in
\Psi^{\mathrm{a}}(Q_{\phi,\,\cdot\,}):xe^{j}=z_{0}\right\} \\
&  =\min_{\phi\in\chi}p_{j}^{\mathrm{a}}(Q_{\phi,\,\cdot\,}),
\end{align*}
where $p_{j}^{\mathrm{a}}(Q_{\phi,\,\cdot\,})$ is the seller's price of an
American option with gradual exercise and payoff process~$Q_{\phi,\,\cdot\,}$,
which can be expressed as%
\[
p_{j}^{\mathrm{a}}(Q_{\phi,\,\cdot\,})=\max_{\psi\in\mathcal{X}}%
\max_{(\mathbb{Q},S)\in\mathcal{\bar{P}}_{j}^{\mathrm{d}}(\psi)}%
\mathbb{E}_{\mathbb{Q}}((Q_{\phi,\,\cdot\,}\cdot S)_{\psi})
\]
by Theorem~\ref{Thm:jjf658wgffvdsdf85}. It follows that%
\[
\pi_{j}^{\mathrm{a}}(Y,X)=\min_{\phi\in\mathcal{X}}\max_{\psi\in\mathcal{X}%
}\max_{(\mathbb{Q},S)\in\mathcal{\bar{P}}_{j}^{\mathrm{d}}(\psi)}%
\mathbb{E}_{\mathbb{Q}}((Q_{\phi,\,\cdot\,}\cdot S)_{\psi}),
\]
completing the proof.\bigskip
\end{proof}

\begin{proof}
[Proof of Theorem~\ref{Thm:fjf75s41s}]Using
Definition~\ref{Def:nfn457w6dbg45d} and Propositions~\ref{Prop:9h75hhsl}
and~\ref{Prop:jg768enk}, we obtain%
\begin{align*}
\pi_{j}^{\mathrm{b}}(Y,X)  &  =\sup\left\{  -x\in\mathbb{R}\,|\,\exists
(\psi,u)\in\Phi^{\mathrm{b}}(Y,X):xe^{j}=u_{0}\right\} \\
&  =\sup\left\{  -x\in\mathbb{R}\,|\,\exists(\psi,u)\in\Lambda^{\mathrm{b}%
}(Y,X):xe^{j}=u_{0}\right\} \\
&  =\sup\left\{  -x\in\mathbb{R}\,|\,xe^{j}\in\mathcal{Z}_{0}^{\mathrm{b}%
}\right\}  .
\end{align*}
Being a polyhedral set, $\mathcal{Z}_{0}^{\mathrm{b}}$ is closed, hence
$\left\{  -x\in\mathbb{R}\,|\,xe^{j}\in\mathcal{Z}_{0}^{\mathrm{b}}\right\}  $
is closed. Moreover, $\left\{  -x\in\mathbb{R}\,|\,xe^{j}\in\mathcal{Z}%
_{0}^{\mathrm{b}}\right\}  $ is non-empty and bounded above because $xe^{j}%
\in\mathcal{Z}_{0}^{\mathrm{b}}$ for any $x\in\mathbb{R}$ large enough and
$xe^{j}\notin\mathcal{Z}_{0}^{\mathrm{b}}$ for any $x\in\mathbb{R}$ small
enough, so the supremum is in fact a maximum.\bigskip
\end{proof}

\begin{proof}
[Proof of Proposition~\ref{Prop:9h75hhsl}]Let $a\in\mathcal{Z}_{0}%
^{\mathrm{b}}$. We construct a mixed stopping time $\psi\in\mathcal{X}$ and a
strategy $z\in\Phi$ by induction. First we put $\psi_{0}^{\ast}:=1$ and
$z_{0}:=a$. Clearly, $z_{0}\in\psi_{0}^{\ast}\mathcal{Z}_{0}^{\mathrm{b}}$.
Next, suppose that $z_{t}\in\psi_{t}^{\ast}\mathcal{Z}_{t}^{\mathrm{b}}$ for
some $t=0,\ldots,T-1$. Then $z_{t}\in\psi_{t}^{\ast}\operatorname*{conv}%
\{\mathcal{V}_{t}^{\mathrm{b}}\cap\mathcal{X}_{t}^{\mathrm{b}},\mathcal{Y}%
_{t}^{\mathrm{b}}\}$, so there exist $\lambda_{t}\in\lbrack0,1]$, $v_{t}%
\in\mathcal{V}_{t}^{\mathrm{b}}\cap\mathcal{X}_{t}^{\mathrm{b}}$ and $y_{t}%
\in\mathcal{Y}_{t}^{\mathrm{b}}$\ such that $z_{t}=\psi_{t}^{\ast}\left(
(1-\lambda_{t})v_{t}+\lambda_{t}y_{t}\right)  $. We put $\psi_{t}:=\psi
_{t}^{\ast}\lambda_{t}$, and then $\psi_{t+1}^{\ast}:=\psi_{t}^{\ast}-\psi
_{t}=\psi_{t}^{\ast}(1-\lambda_{t})$, so $z_{t}=\psi_{t+1}^{\ast}v_{t}%
+\psi_{t}y_{t}$. Because $v_{t}\in\mathcal{X}_{t}^{\mathrm{b}}$ and $y_{t}%
\in\mathcal{Y}_{t}^{\mathrm{b}}$, we have $v_{t}+X_{t}\in\mathcal{Q}_{t}$ and
$y_{t}+Y_{t}\in\mathcal{Q}_{t}$. It follows that%
\begin{align*}
z_{t}+\psi_{t}Y_{t}+\psi_{t+1}^{\ast}X_{t}  &  =\psi_{t+1}^{\ast}v_{t}%
+\psi_{t}y_{t}+\psi_{t}Y_{t}+\psi_{t+1}^{\ast}X_{t}\\
&  =\psi_{t+1}^{\ast}(v_{t}+X_{t})+\psi_{t}(y_{t}+Y_{t})\\
&  \in\psi_{t+1}^{\ast}\mathcal{Q}_{t}+\psi_{t}\mathcal{Q}_{t}\subseteq
\mathcal{Q}_{t}.
\end{align*}
Since $v_{t}\in\mathcal{V}_{t}^{\mathrm{b}}$, there is a $z_{t+1}\in\psi
_{t+1}^{\ast}\mathcal{W}_{t}^{\mathrm{b}}$ such that $\psi_{t+1}^{\ast}%
v_{t}-z_{t+1}\in\mathcal{Q}_{t}$. It follows that%
\begin{align*}
z_{t}+\psi_{t}Y_{t}-z_{t+1}  &  =\psi_{t+1}^{\ast}v_{t}+\psi_{t}y_{t}+\psi
_{t}Y_{t}-z_{t+1}\\
&  =(\psi_{t+1}^{\ast}v_{t}-z_{t+1})+\psi_{t}(y_{t}+Y_{t})\in\mathcal{Q}_{t}.
\end{align*}
Since $\mathcal{W}_{t}^{\mathrm{b}}\subseteq\mathcal{Z}_{t+1}^{\mathrm{b}}$,
it also follows that $z_{t+1}\in\psi_{t+1}^{\ast}\mathcal{Z}_{t+1}%
^{\mathrm{b}}$. Finally, given that $z_{T}\in\psi_{T}^{\ast}\mathcal{Z}%
_{T}^{\mathrm{b}}=\psi_{T}^{\ast}\mathcal{Y}_{T}^{\mathrm{b}}$, we get
$z_{T}+\psi_{T}Y_{T}=z_{T}+\psi_{T}^{\ast}Y_{T}\in\mathcal{Q}_{T}$. Putting
and $\psi_{T}:=\psi_{T}^{\ast}$, $\psi_{T+1}^{\ast}:=0$ and $z_{T+1}:=0$, we
obtain%
\[
z_{T}+\psi_{T}Y_{T}+\psi_{T+1}^{\ast}X_{T}\in\mathcal{Q}_{T}.
\]
We have constructed $(\psi,z)\in\Lambda^{\mathrm{b}}(Y,X)$ such that $a=z_{0}$.

Conversely, we take any $(\psi,z)\in\Lambda^{\mathrm{b}}(Y,X)$, and want to
show that $z_{0}\in\mathcal{Z}_{0}^{\mathrm{b}}$. This a consequence of the
following fact, which will be proved by backward induction: for each
$t=0,\ldots,T$
\begin{equation}
z_{t}\in\left\{
\begin{array}
[c]{ll}%
\psi_{t}^{\ast}\mathcal{Z}_{t}^{\mathrm{b}} & \text{on }\{\psi_{t}^{\ast}>0\}\\
\mathcal{Q}_{t} & \text{on }\{\psi_{t}^{\ast}=0\}
\end{array}
\right.  . \label{eq:fh547sh4}%
\end{equation}
We start the proof with $t=T$. Since $z_{T}+\psi_{T}^{\ast}Y_{T}=z_{T}%
+\psi_{T}Y_{T}+\psi_{T+1}^{\ast}X_{T}\in\mathcal{Q}_{T}$, it follows that
indeed
\[
z_{T}\in-\psi_{T}^{\ast}Y_{T}+\mathcal{Q}_{T}=\left\{
\begin{array}
[c]{ll}%
\psi_{T}^{\ast}\mathcal{Y}_{T}^{\mathrm{b}} & \text{on }\{\psi_{T}^{\ast}>0\}\\
\mathcal{Q}_{T} & \text{on }\{\psi_{T}^{\ast}=0\}
\end{array}
\right.  =\left\{
\begin{array}
[c]{ll}%
\psi_{T}^{\ast}\mathcal{Z}_{T}^{\mathrm{b}} & \text{on }\{\psi_{T}^{\ast}>0\}\\
\mathcal{Q}_{T} & \text{on }\{\psi_{T}^{\ast}=0\}
\end{array}
\right.  .
\]
Next, suppose that (\ref{eq:fh547sh4}) holds for some $t=1,\ldots,T$. Since
$z$ is predictable, it follows that $z_{t}\in\mathcal{L}_{t-1}$, so%
\[
z_{t}\in\left\{
\begin{array}
[c]{ll}%
\psi_{t}^{\ast}(\mathcal{Z}_{t}^{\mathrm{b}}\cap\mathcal{L}_{t-1}) & 
\text{on }\{\psi_{t}^{\ast}>0\}\\
\mathcal{Q}_{t}\cap\mathcal{L}_{t-1} & \text{on }\{\psi_{t}^{\ast}=0\}
\end{array}
\right.  \subseteq\left\{
\begin{array}
[c]{ll}%
\psi_{t}^{\ast}\mathcal{W}_{t-1}^{\mathrm{b}} & \text{on }\{\psi_{t}^{\ast}>0\}\\
\mathcal{Q}_{t-1} & \text{on }\{\psi_{t}^{\ast}=0\}
\end{array}
\right.  .
\]
Hence, using%
\[
z_{t-1}+\psi_{t-1}Y_{t-1}-z_{t}\in\mathcal{Q}_{t-1},
\]
we obtain%
\setlongestterm{\psi_{t}^{\ast}\mathcal{W}_{t-1}^{\mathrm{b}}+\mathcal{Q}_{t-1}}
\begin{align*}
z_{t-1}+\psi_{t-1}Y_{t-1}  &  \in z_{t}+\mathcal{Q}_{t-1}\\
&  \subseteq\left\{
\begin{array}
[c]{ll}%
\psi_{t}^{\ast}\mathcal{W}_{t-1}^{\mathrm{b}}+\mathcal{Q}_{t-1} & 
\text{on }\{\psi_{t}^{\ast}>0\}\\
\mathcal{Q}_{t-1}+\mathcal{Q}_{t-1} & \text{on }\{\psi_{t}^{\ast}=0\}
\end{array}
\right. \\
&  =\left\{
\begin{array}
[c]{ll}%
\adjusttolongestterm{\psi_{t}^{\ast}\mathcal{V}_{t-1}^{\mathrm{b}}} & \text{on }\{\psi_{t}^{\ast}>0\}\\
\mathcal{Q}_{t-1} & \text{on }\{\psi_{t}^{\ast}=0\}
\end{array}
\right.  .
\end{align*}
Moreover, since%
\[
z_{t-1}+\psi_{t-1}Y_{t-1}+\psi_{t}^{\ast}X_{t-1}\in\mathcal{Q}_{t-1},
\]
we obtain%
\[
z_{t-1}+\psi_{t-1}Y_{t-1}\in\mathcal{Q}_{t-1}-\psi_{t}^{\ast}X_{t-1}=\left\{
\begin{array}
[c]{ll}%
\psi_{t}^{\ast}\mathcal{X}_{t-1}^{\mathrm{b}} & \text{on }\{\psi_{t}^{\ast}>0\}\\
\mathcal{Q}_{t-1} & \text{on }\{\psi_{t}^{\ast}=0\}
\end{array}
\right.  .
\]
It follows that%
\[
z_{t-1}+\psi_{t-1}Y_{t-1}\in\left\{
\begin{array}
[c]{ll}%
\psi_{t}^{\ast}(\mathcal{V}_{t-1}^{\mathrm{b}}\cap\mathcal{X}_{t-1}%
^{\mathrm{b}}) & \text{on }\{\psi_{t}^{\ast}>0\}\\
\mathcal{Q}_{t-1} & \text{on }\{\psi_{t}^{\ast}=0\}
\end{array}
\right.  ,
\]
and so%
\setlongestterm{\psi_{t-1}^{\ast}\operatorname*{conv}\left\{  \mathcal{V}_{t-1}^{\mathrm{b}%
}\cap\mathcal{X}_{t-1}^{\mathrm{b}},\mathcal{Y}_{t-1}^{\mathrm{b}}\right\}}
\begin{align*}
z_{t-1}  &  \in\left\{
\begin{array}
[c]{ll}%
\adjusttolongestterm[+\innotis]{\psi_{t}^{\ast}(\mathcal{V}_{t-1}^{\mathrm{b}}\cap\mathcal{X}_{t-1}%
^{\mathrm{b}})+\psi_{t-1}Y_{t-1}} & \text{on }\{\psi_{t}^{\ast}>0\}\\
\mathcal{Q}_{t-1}+\psi_{t-1}Y_{t-1} & \text{on }\{\psi_{t}^{\ast}=0\}
\end{array}
\right. \\
&  \subseteq\left\{
\begin{array}
[c]{ll}%
\adjusttolongestterm{\psi_{t}^{\ast}(\mathcal{V}_{t-1}^{\mathrm{b}}\cap\mathcal{X}_{t-1}%
^{\mathrm{b}})+\psi_{t-1}\mathcal{Y}_{t-1}^{\mathrm{b}}} & \text{on }\{\psi
_{t}^{\ast}>0\}\\
\mathcal{Q}_{t-1}+\psi_{t-1}\mathcal{Y}_{t-1}^{\mathrm{b}} & \text{on }\{%
\psi_{t}^{\ast}=0\}
\end{array}
\right. \\
&  =\left\{
\begin{array}
[c]{ll}%
\adjusttolongestterm[-\innotis]{\psi_{t}^{\ast}(\mathcal{V}_{t-1}^{\mathrm{b}}\cap\mathcal{X}_{t-1}%
^{\mathrm{b}})+\psi_{t-1}\mathcal{Y}_{t-1}^{\mathrm{b}}} & \text{on }\{\psi
_{t}^{\ast}>0\}\\
\psi_{t-1}\mathcal{Y}_{t-1}^{\mathrm{b}} & \text{on }\{\psi_{t}^{\ast}=0,
\psi_{t-1}^{\ast}>0\}\\
\mathcal{Q}_{t-1} & \text{on }\{\psi_{t-1}^{\ast}=0\}
\end{array}
\right. \\
&  =\left\{
\begin{array}
[c]{ll}%
\adjusttolongestterm{\psi_{t}^{\ast}(\mathcal{V}_{t-1}^{\mathrm{b}}\cap\mathcal{X}_{t-1}%
^{\mathrm{b}})+\psi_{t-1}\mathcal{Y}_{t-1}^{\mathrm{b}}} & \text{on }\{\psi
_{t-1}^{\ast}>0\}\\
\mathcal{Q}_{t-1} & \text{on }\{\psi_{t-1}^{\ast}=0\}
\end{array}
\right. \\
&  \subseteq\left\{
\begin{array}
[c]{ll}%
\psi_{t-1}^{\ast}\operatorname*{conv}\left\{  \mathcal{V}_{t-1}^{\mathrm{b}%
}\cap\mathcal{X}_{t-1}^{\mathrm{b}},\mathcal{Y}_{t-1}^{\mathrm{b}}\right\}  &
\text{on }\{\psi_{t-1}^{\ast}>0\}\\
\mathcal{Q}_{t-1} & \text{on }\{\psi_{t-1}^{\ast}=0\}
\end{array}
\right. \\
&  =\left\{
\begin{array}
[c]{ll}%
\adjusttolongestterm{\psi_{t-1}^{\ast}\mathcal{Z}_{t-1}^{\mathrm{b}}} & \text{on }\{\psi_{t-1}^{\ast}>0\}\\
\mathcal{Q}_{t-1} & \text{on }\{\psi_{t-1}^{\ast}=0\}
\end{array}
,\right.
\end{align*}
which completes the proof.\bigskip
\end{proof}

\begin{proof}
[Proof of Proposition~\ref{Prop:jg768enk}]Suppose that $(\psi,z)\in
\Lambda^{\mathrm{b}}(Y,X)$. Then, for each $t=0,\ldots,T-1$,%
\[
z_{t}+\psi_{t}Y_{t}-z_{t+1}\in\mathcal{Q}_{t},
\]
so there is a liquidation strategy $y_{t+1}^{t},\ldots,y_{T+1}^{t}$ starting
from $z_{t}+\psi_{t}Y_{t}-z_{t+1}$ at time~$t$. We also put $y_{T+1}^{T}:=0$
for notational convenience. Moreover,
\[
z_{t}+\psi_{t}Y_{t}+\psi_{t+1}^{\ast}X_{t}\in\mathcal{Q}_{t}\quad\text{for
each }t=0,\ldots,T,
\]
so there is a liquidation strategy $x_{t+1}^{t},\ldots,x_{T+1}^{t}$ starting
from $z_{t}+\psi_{t}Y_{t}+\psi_{t+1}^{\ast}X_{t}$ at time~$t$. For each
$\phi\in\mathcal{X}$ we put%
\begin{align*}
u_{0}^{\phi}  &  :=z_{0},\\
u_{t}^{\phi}  &  :=\phi_{t}^{\ast}z_{t}+\sum_{s=0}^{t-1}\phi_{s+1}^{\ast}%
y_{t}^{s}+\sum_{s=0}^{t-1}\phi_{s}x_{t}^{s}\quad\text{for }t=1,\ldots,T+1.
\end{align*}
This defines $u:\mathcal{X}\rightarrow\Phi$, which satisfies the
non-anticipation condition~(\ref{eq:dhg3sbdj}). Moreover, for each $\psi
\in\mathcal{X}$ and for each $t=0,\ldots,T$%
\begin{align*}
&  u_{t}^{\phi}+G_{t}^{\phi,\psi}-u_{t+1}^{\phi}\\
&  =\phi_{t}^{\ast}z_{t}+\sum_{s=0}^{t-1}\phi_{s+1}^{\ast}y_{t}^{s}+\sum
_{s=0}^{t-1}\phi_{s}x_{t}^{s}+\psi_{t}\phi_{t}^{\ast}Y_{t}+\psi_{t+1}^{\ast
}\phi_{t}X_{t}\\
&  \quad-\phi_{t+1}^{\ast}z_{t+1}-\sum_{s=0}^{t}\phi_{s+1}^{\ast}y_{t+1}%
^{s}-\sum_{s=0}^{t}\phi_{s}x_{t+1}^{s}\\
&  =\phi_{t+1}^{\ast}\left(  z_{t}+\psi_{t}Y_{t}-z_{t+1}-y_{t+1}^{t}\right)
+\phi_{t}\left(  z_{t}+\psi_{t}Y_{t}+\psi_{t+1}^{\ast}X_{t}-x_{t+1}^{t}\right)
\\
&  \quad+\sum_{s=0}^{t-1}\phi_{s+1}^{\ast}\left(  y_{t}^{s}-y_{t+1}%
^{s}\right)  +\sum_{s=0}^{t-1}\phi_{s}\left(  x_{t}^{s}-x_{t+1}^{s}\right) \\
&  \in\phi_{t+1}^{\ast}\mathcal{K}_{t}+\phi_{t}\mathcal{K}_{t}+\sum
_{s=0}^{t-1}\phi_{s+1}^{\ast}\mathcal{K}_{t}+\sum_{s=0}^{t-1}\phi
_{s}\mathcal{K}_{t}\subseteq\mathcal{K}_{t}.
\end{align*}
This means that $(\psi,u)\in\Phi^{\mathrm{b}}(Y,X)$, with $z_{0}=u_{0}$.

Conversely, suppose that $(\psi,u)\in\Phi^{\mathrm{b}}(Y,X)$. Then we put%
\[
z:=u^{\chi^{T}}.
\]
It follows that for each $t=0,\ldots,T-1$%
\[
z_{t}+\psi_{t}Y_{t}-z_{t+1}=u_{t}^{\chi^{T}}-G_{t}^{\chi^{T},\psi}%
-u_{t+1}^{\chi^{T}}\in\mathcal{K}_{t}\subseteq\mathcal{Q}_{t}%
\]
since%
\[
G_{t}^{\chi^{T},\psi}=\psi_{t}\chi_{t}^{T\ast}Y_{t}+\psi_{t+1}^{\ast}\chi
_{t}^{T}X_{t}=\psi_{t}Y_{t}.
\]
Next, take any $t=0,\ldots,T$. Then $\chi_{s}^{T}=\chi_{s}^{t}=0$ for each
$s=0,\ldots,t-1$, and because $u$ satisfies the non-anticipation
condition~(\ref{eq:dhg3sbdj}), we have $z_{t}=u_{t}^{\chi^{T}}=u_{t}^{\chi
^{t}}$. Since $\chi_{t}^{t}=\chi_{t}^{t\ast}=1$ and%
\[
G_{t}^{\chi^{t},\psi}=\psi_{t}\chi_{t}^{t\ast}Y_{t}+\psi_{t+1}^{\ast}\chi
_{t}^{t}X_{t}=\psi_{t}Y_{t}+\psi_{t+1}^{\ast}X_{t},
\]
it means that%
\begin{equation}
z_{t}+\psi_{t}Y_{t}+\psi_{t+1}^{\ast}X_{t}-u_{t+1}^{\chi^{t}}=u_{t}^{\chi^{t}%
}+G_{t}^{\chi^{t},\psi}-u_{t+1}^{\chi^{t}}\in\mathcal{K}_{t}\subseteq
\mathcal{Q}_{t}. \label{eq:hdta5347}%
\end{equation}
Moreover, for each $s=t+1,\ldots,T$ we have $\chi_{s}^{t}=\chi_{s}^{t\ast}=0$
and%
\[
G_{s}^{\chi^{t},s}=\psi_{s}\chi_{s}^{t\ast}Y_{s}+\psi_{s+1}^{\ast}\chi_{s}%
^{t}X_{s}=0,
\]
hence%
\[
u_{s}^{\chi^{t}}-u_{s+1}^{\chi^{t}}=u_{s}^{\chi^{t}}+G_{s}^{\chi^{t},\psi
}-u_{s+1}^{\chi^{t}}\in\mathcal{K}_{s}\subseteq\mathcal{Q}_{s}.
\]
We can verify by backward induction that $u_{s+1}^{\chi^{t}}\in\mathcal{Q}%
_{s}$ for each $s=t,\ldots,T$. Clearly, $u_{T+1}^{\chi^{t}}=0\in
\mathcal{Q}_{T}$. Now suppose that $u_{s+1}^{\chi^{t}}\in\mathcal{Q}_{s}$ for
some $s=t+1,\ldots,T$. It follows that $u_{s}^{\chi^{t}}=(u_{s}^{\chi^{t}%
}-u_{s+1}^{\chi^{t}})+u_{s+1}^{\chi^{t}}\in\mathcal{Q}_{s}+\mathcal{Q}%
_{s}=\mathcal{Q}_{s}$. By predictability, $u_{s}^{\chi^{t}}\in\mathcal{L}%
_{s-1}$, so we can conclude that $u_{s}^{\chi^{t}}\in\mathcal{Q}_{s}%
\cap\mathcal{L}_{s-1}\subseteq\mathcal{Q}_{s-1}$, which completes the backward
induction argument. In particular, we have shown that $u_{t+1}^{\chi^{t}}%
\in\mathcal{Q}_{t}$. Together with~(\ref{eq:hdta5347}), this shows that%
\[
z_{t}+\psi_{t}Y_{t}+\psi_{t+1}^{\ast}X_{t}\in\mathcal{Q}_{t}%
\]
for each $t=0,\ldots,T$. As a result, $(\psi,z)\in\Lambda^{\mathrm{b}}(Y,X)$
with $z_{0}=u_{0}$, which completes the proof.\bigskip
\end{proof}

\begin{proof}
[Proof of Lemma~\ref{Lem:k856dg3a}]Observe that for any $\psi\in\mathcal{X}$%
\begin{align*}
Q_{t,\psi}  &  =Q_{\chi^{t},\psi}=\sum_{u=0}^{T}\sum_{s=0}^{T}\chi_{u}^{t}%
\psi_{s}Q_{u,s}=\sum_{s=0}^{T}\psi_{s}\left(  \mathbf{1}_{\left\{  t\geq
s\right\}  }Y_{s}+\mathbf{1}_{\left\{  t<s\right\}  }X_{t}\right) \\
&  =\sum_{s=0}^{t}\psi_{s}Y_{s}+\psi_{t+1}^{\ast}X_{t}.
\end{align*}
Now take any $u\in\Phi$ such that $(\psi,u)\in\Phi^{\mathrm{b}}(Y,X)$ and
define $z:\mathcal{X}\rightarrow\Phi$ such that%
\begin{equation}
z_{t}^{\phi}:=u_{t}^{\phi}-\phi_{t}^{\ast}\sum_{s=0}^{t-1}\psi_{s}Y_{s}
\label{eq:ff645lallpr}%
\end{equation}
for any $\phi\in\mathcal{X}$ and any $t=0,\ldots,T+1$.\ Then~$z$ satisfies the
non-anticipation condition~(\ref{eq:dg453hd7ah}), $z_{0}=u_{0}$, and it also
satisfies the rebalancing condition~(\ref{eq:jgt775g10}) for an American
option with payoff process~$Z=Q_{\,\cdot\,,\psi}$ since for any $t=0,\ldots,T$%
\begin{align*}
&  z_{t}^{\phi}+\phi_{t}Q_{t,\psi}-z_{t+1}^{\phi}\\
&  =\left(  u_{t}^{\phi}-\phi_{t}^{\ast}\sum_{s=0}^{t-1}\psi_{s}Y_{s}\right)
+\phi_{t}\left(  \sum_{s=0}^{t}\psi_{s}Y_{s}+\psi_{t+1}^{\ast}X_{t}\right)
-\left(  u_{t+1}^{\phi}-\phi_{t+1}^{\ast}\sum_{s=0}^{t}\psi_{s}Y_{s}\right) \\
&  =u_{t}^{\phi}+\psi_{t}\phi_{t}^{\ast}Y_{t}+\psi_{t+1}^{\ast}\phi_{t}%
X_{t}-u_{t+1}^{\phi}\\
&  =u_{t}^{\phi}+G_{t}^{\phi,\psi}-u_{t+1}^{\phi}\in\mathcal{K}_{t}.
\end{align*}
Conversely, take any $z\in\Psi^{\mathrm{a}}(-Q_{\,\cdot\,,\psi})$ and define
$u:\mathcal{X}\rightarrow\Phi$ such that%
\[
u_{t}^{\phi}:=z_{t}^{\phi}+\phi_{t}^{\ast}\sum_{s=0}^{t-1}\psi_{s}Y_{s}%
\]
for any $\phi\in\mathcal{X}$ and any $t=0,\ldots,T+1$. Then $u$ satisfies the
non-anticipation condition~(\ref{eq:dhg3sbdj}), $u_{0}=z_{0}$, and
\begin{align*}
&  u_{t}^{\phi}+G_{t}^{\phi,\psi}-u_{t+1}^{\phi}\\
&  =\left(  z_{t}^{\phi}+\phi_{t}^{\ast}\sum_{s=0}^{t-1}\psi_{s}Y_{s}\right)
+\left(  \psi_{t}\phi_{t}^{\ast}Y_{t}+\psi_{t+1}^{\ast}\phi_{t}X_{t}\right)
-\left(  z_{t+1}^{\phi}+\phi_{t+1}^{\ast}\sum_{s=0}^{t}\psi_{s}Y_{s}\right) \\
&  =z_{t}^{\phi}+\phi_{t}\left(  \sum_{s=0}^{t}\psi_{s}Y_{s}+\psi_{t+1}^{\ast
}X_{t}\right)  -z_{t+1}^{\phi}\\
&  =z_{t}^{\phi}+\phi_{t}Q_{t,\psi}-z_{t+1}^{\phi}\in\mathcal{K}_{t}%
\end{align*}
for any $\phi\in\mathcal{X}$ and any $t=0,\ldots,T$, that is, the rebalancing
condition~(\ref{eq:hf65hsdn9}) holds. The lemma follows because
(\ref{eq:ff645lallpr}) defines a one-to-one map between strategies $z\in
\Psi^{\mathrm{a}}(-Q_{\,\cdot\,,\psi})$ and strategies $u$ such that
$(\psi,u)\in\Phi^{\mathrm{b}}(Y,X)$ with $u_{0}=z_{0}$.\bigskip
\end{proof}

\begin{proof}
[Proof of Theorem~\ref{Thm:dhhd5436syy}]By Definition~\ref{Def:nfn457w6dbg45d}%
,%
\begin{equation}
\pi_{j}^{\mathrm{b}}(Y,X)=\sup\left\{  -x\in\mathbb{R}\,|\,\exists(\psi
,u)\in\Phi^{\mathrm{b}}(Y,X):xe^{j}=u_{0}\right\}  . \label{Eq:hfhtu5hnla}%
\end{equation}
According to Theorem~\ref{Thm:fjf75s41s}, $-\pi_{j}^{\mathrm{b}}(Y,X)e^{j}%
\in\mathcal{Z}_{0}^{\mathrm{b}}$. Hence, by Proposition~\ref{Prop:9h75hhsl},
there is a $(\psi,z)\in\Lambda^{\mathrm{b}}(Y,X)$ such that $-\pi
_{j}^{\mathrm{b}}(Y,X)e^{j}=z_{0}$, and so, by Proposition~\ref{Prop:jg768enk}%
, there is a $(\psi,u)\in\Phi^{\mathrm{b}}(Y,X)$ such that $-\pi
_{j}^{\mathrm{b}}(Y,X)e^{j}=u_{0}$. It follows that the supremum
in~(\ref{Eq:hfhtu5hnla}) is attained,%
\begin{align*}
\pi_{j}^{\mathrm{b}}(Y,X)  &  =\max\left\{  -x\in\mathbb{R}\,|\,\exists
(\psi,u)\in\Phi^{\mathrm{b}}(Y,X):xe^{j}=u_{0}\right\} \\
&  =-\min\left\{  x\in\mathbb{R}\,|\,\exists(\psi,u)\in\Phi^{\mathrm{b}%
}(Y,X):xe^{j}=u_{0}\right\}  .
\end{align*}
Hence, according to Lemma~\ref{Lem:k856dg3a} and
Definition~\ref{Def:hfy468ahd564vk},
\begin{align*}
\pi_{j}^{\mathrm{b}}(Y,X)  &  =-\min\left\{  x\in\mathbb{R}\,|\,\exists\psi
\in\mathcal{X}\exists z\in\Psi^{\mathrm{a}}(-Q_{\,\cdot\,,\psi}):xe^{j}%
=z_{0}\right\} \\
&  =-\min_{\psi\in\mathcal{X}}\inf\left\{  x\in\mathbb{R}\,|\,\exists z\in
\Psi^{\mathrm{a}}(-Q_{\,\cdot\,,\psi}):xe^{j}=z_{0}\right\} \\
&  =-\min_{\psi\in\mathcal{X}}p_{j}^{\mathrm{a}}(-Q_{\,\cdot\,,\psi}),
\end{align*}
where
\[
p_{j}^{\mathrm{a}}(-Q_{\,\cdot\,,\psi})=\inf\left\{  x\in\mathbb{R}%
\,|\,\exists z\in\Psi^{\mathrm{a}}(-Q_{\,\cdot\,,\psi}):xe^{j}=z_{0}\right\}
\]
is the seller's price of an American option under gradual exercise with payoff
process~$-Q_{\,\cdot\,,\psi}$. By Theorem~\ref{Thm:jjf658wgffvdsdf85},%
\[
p_{j}^{\mathrm{a}}(-Q_{\,\cdot\,,\psi})=\max_{\phi\in\mathcal{X}}%
\max_{(\mathbb{Q},S)\in\mathcal{\bar{P}}_{j}^{\mathrm{d}}(\phi)}%
\mathbb{E}_{\mathbb{Q}}(\left(  -Q_{\,\cdot\,,\psi}\cdot S\right)  _{\phi}),
\]
so%
\begin{align*}
\pi_{j}^{\mathrm{b}}(Y,X)  &  =-\min_{\psi\in\mathcal{X}}p_{j}^{\mathrm{a}%
}(-Q_{\,\cdot\,,\psi})\\
&  =-\min_{\psi\in\mathcal{X}}\max_{\phi\in\mathcal{X}}\max_{(\mathbb{Q}%
,S)\in\mathcal{\bar{P}}_{j}^{\mathrm{d}}(\phi)}\mathbb{E}_{\mathbb{Q}}(\left(
-Q_{\,\cdot\,,\psi}\cdot S\right)  _{\phi})\\
&  =\max_{\psi\in\mathcal{X}}\min_{\phi\in\mathcal{X}}\min_{(\mathbb{Q}%
,S)\in\mathcal{\bar{P}}_{j}^{\mathrm{d}}(\phi)}\mathbb{E}_{\mathbb{Q}}(\left(
Q_{\,\cdot\,,\psi}\cdot S\right)  _{\phi}).
\end{align*}
Theorem~\ref{Thm:dhhd5436syy} has been proved.
\end{proof}

\end{document}